\documentclass[letterpaper,twocolumn,american,showpacs,pr,aps,superscriptaddress,eqsecnum,nofootinbib]{revtex4}

\usepackage{amsfonts}%
\usepackage{amsmath}%
\usepackage{amssymb}%
\usepackage{graphicx}
\usepackage{color}
\usepackage{lipsum}
\usepackage{comment}
\usepackage{hyperref}
\providecommand{\U}[1]{\protect\rule{.1in}{.1in}}
\newtheorem{theorem}{Theorem}

\newtheorem{corollary}[theorem]{Corollary}

\newtheorem{definition}[theorem]{Definition}
\newtheorem{example}[theorem]{Example}

\newtheorem{lemma}[theorem]{Lemma}

\newtheorem{proposition}[theorem]{Proposition}

\newenvironment{proof}[1][Proof]{\noindent\textbf{#1.} }{\ \rule{0.5em}{0.5em}}

\definecolor{nblue}{rgb}{0.2,0.2,0.7}
\definecolor{ngreen}{rgb}{0.2,0.6,0.2}
\definecolor{nred}{rgb}{0.7,0.2,0.2}
\definecolor{nblack}{rgb}{0,0,0}

\begin{document}

\title{Modes of asymmetry: the application of harmonic analysis to symmetric quantum dynamics and quantum reference frames}

\author{Iman Marvian}
\affiliation{Perimeter Institute for Theoretical Physics, 31 Caroline St. N, Waterloo, \\
Ontario, Canada N2L 2Y5}
\affiliation{Institute for Quantum Computing, University of Waterloo, 200 University Ave. W, Waterloo, Ontario, Canada N2L 3G1}
\affiliation{Department of Physics and Astronomy, Center for Quantum Information Science and Technology, University of Southern California, Los Angeles, CA 90089}

\author{Robert W. Spekkens}
\affiliation{Perimeter Institute for Theoretical Physics, 31 Caroline St. N, Waterloo, \\
Ontario, Canada N2L 2Y5}

\date{Oct. 28, 2013}

\begin{abstract}
Finding the consequences of symmetry for open system quantum dynamics is a problem with broad applications, including describing thermal relaxation, deriving quantum limits on the performance of amplifiers, and exploring quantum metrology in the presence of noise.  The symmetry of the dynamics may reflect a symmetry of the fundamental laws of nature, a symmetry of a low-energy effective theory, or it may describe a practical restriction such as the lack of a reference frame. In this paper, we apply some tools of harmonic analysis together with ideas from quantum information theory to   this problem. 
The central idea is to study the decomposition of quantum operations---in particular, states, measurements and channels---into different modes, which we call \emph{modes of asymmetry}.
Under symmetric processing, a given mode of the input is mapped to the corresponding mode of the output, implying that one can only generate a given output if the input contains all of the necessary modes.
By defining monotones that quantify the asymmetry in a particular mode, we also derive 
quantitative constraints on the resources of asymmetry that are required to simulate a given asymmetric operation.
We present applications of our results for deriving bounds on the probability of success in nondeterministic state transitions, such as quantum amplification, and a simplified formalism for studying the degradation of quantum reference frames. 

\end{abstract}

\maketitle

\tableofcontents

\newpage

\section{Introduction}

Extracting non-trivial information about a system's dynamics based on its symmetries is a standard technique in physics. Noether's theorem is a prime example: it allows one to infer conservation laws from symmetries of closed-system dynamics.  
As it turns out however, for \emph{mixed} quantum states, the Noether conservation laws do not capture all of the constraints on state transitions that arise from symmetries.   Furthermore, for open-system dynamics, there are nontrivial constraints on state transitions arising from symmetries even though Noether's theorem does not imply any~\cite{MarvianSpekkensNoether}.

The symmetry of a closed-system dynamics is simply the symmetry of the Hamiltonian that describes the dynamics. 
In an open-system dynamics, the system is not isolated but interacts with its environment.  Then, if the total Hamiltonian of the system and  environment respects a symmetry, and furthermore if the  initial state of the environment also respects that symmetry then the effective evolution of the system will also have that symmetry.  

Finding the set of \emph{all} constraints on state transitions that are implied by symmetries of the dynamics (open or closed) is an important open problem. Solving it motivates the development of a general theory of the \emph{asymmetry properties} of states, that is, the properties which describe the manner in which a state \emph{breaks} symmetries, because these are the properties that determine the possibility of state transitions under symmetric dynamics.

Developing such a theory is also important for the study of quantum references frames (see \cite{BRS07} for a review).  In a context where the only experimental operations that can be freely implemented by an agent are symmetric, an asymmetric state becomes a resource because it can be used to simulate asymmetric channels and asymmetric measurements~\cite{BRST06,BRST06b,Poulin-Yard,Sher-Bart,Mar-Man,BRST04,BIM,Ahm-Rud,MarvianSpekkensWAY} (See Fig.~\ref{Fig:Simulation}).  The restriction to symmetric operations can be understood as the result of lacking a reference frame, and an asymmetric state can be understood as a quantum token of the missing reference frame, allowing the agent to simulate measurements and transformations that are defined relative to the frame.  So this provides another motivation for developing a general theory of asymmetry, one that characterizes not only the asymmetry properties of states, but of channels and measurements as well.

\begin{figure}[h!]
 \center{   \includegraphics[width=6cm]{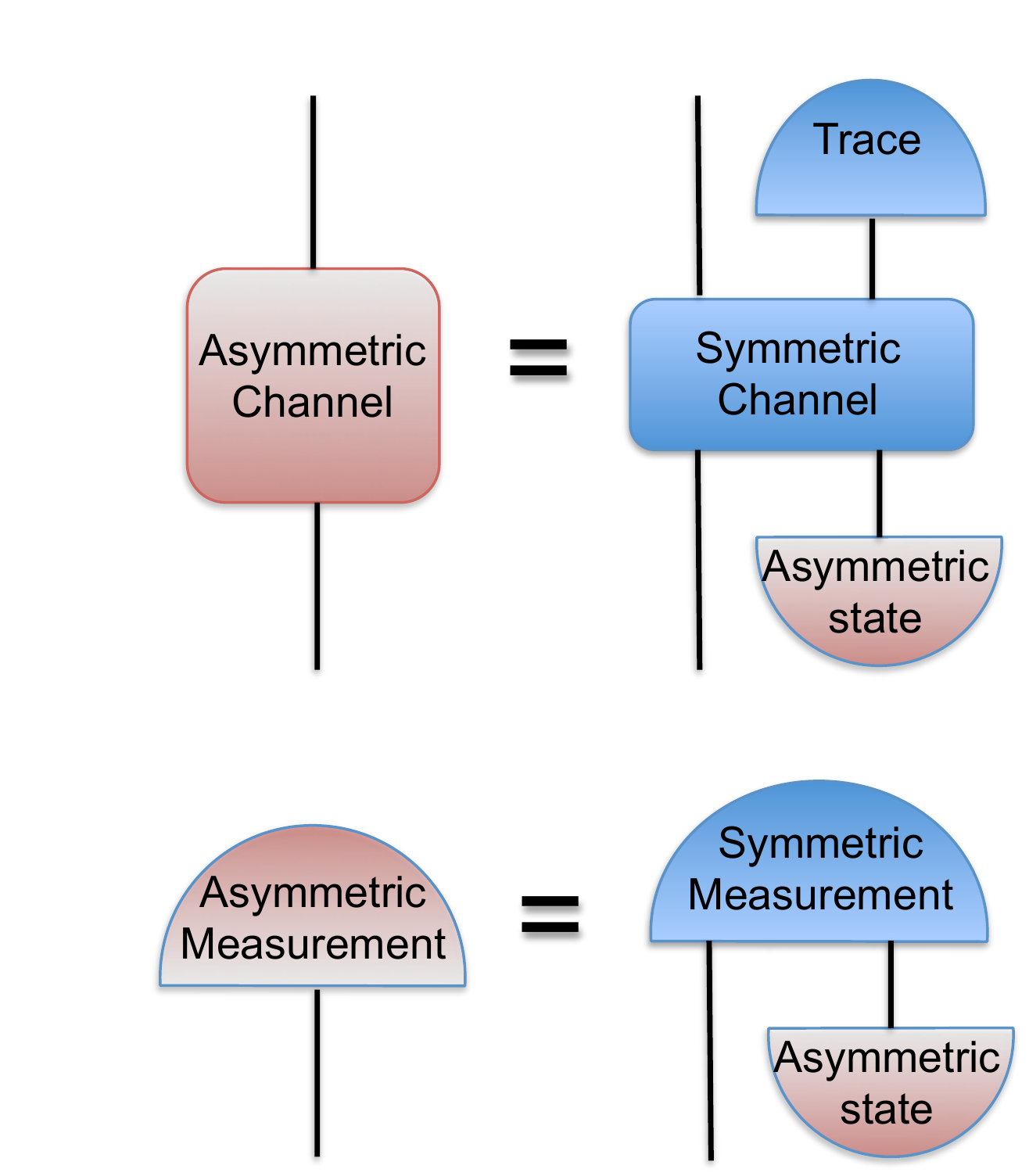}}
\caption{\label{Fig:Simulation}Circuit diagrams depicting how an asymmetric state on one system can be used together with symmetric operations to simulate an asymmetric channel or an asymmetric measurement on another system.}
\end{figure}

There has been significant progress towards this goal in recent years, in particular, on the asymmetry properties of \emph{pure} states~\cite{GS07, GMS09, MS11, MS11-Short, thesis}. For instance, Ref.~\cite{MS11} provides a characterization of the equivalence classes of pure asymmetric states and the necessary and sufficient conditions for one pure state to be converted to another by symmetric processing for any symmetry corresponding to a compact Lie group.  In the case of general mixed states, however, the problem is much harder and much less is known. 
Furthermore, there has been very little work on developing a unified framework for characterizing the asymmetry properties of quantum channels and measurements for arbitrary symmetry groups.

The theory of asymmetry also provides a framework for understanding quantum coherence as a resource. Coherence is considered in many cases to be the signature of quantum behaviour.
Famous quantum phenomena such as the wave nature of particles, superconductivity, and superfluidity can all be interpreted as manifestations of quantum coherence.  To understand the relation between asymmetry and coherence, consider the following example from quantum optics.  Suppose $|n\rangle$ is  the state with $n$ photons in a given mode, and $|0\rangle$ is the vacuum state.  Consider the coherent superposition $\frac{1}{\sqrt{2}}(|0\rangle+|n\rangle)$ and the incoherent mixture $\frac{1}{{2}}(|0\rangle\langle 0|+|n\rangle\langle n|)$.   
One way to understand the  difference between these two states is that the coherent superposition is sensitive to phase shifts while the incoherent  mixture is not.  As such, coherence can be defined as asymmetry relative to phase shifts.  This connection is explored further in Appendix~\ref{App:coherence}.


This article develops the theory of asymmetry by focusing on Fourier decompositions of quantum states, quantum measurements and quantum channels.  

To give the flavour of our approach, we begin by recalling the significance of harmonic analysis (equivalently, Fourier analysis) for classical signal processing.  If a processing is both linear and symmetric under time-translations (one says that the system is \emph{linear time invariant} in this case), then one can decompose the input and output signals to different Fourier modes, i.e. different frequencies, such that 
a signal with frequency $\omega$ at the input can only generate a signal with the same frequency $\omega$ at the output (See Fig.~\ref{Fig:LTI}). 

\begin{figure}[h!]
 \center{   \includegraphics[width=8cm]{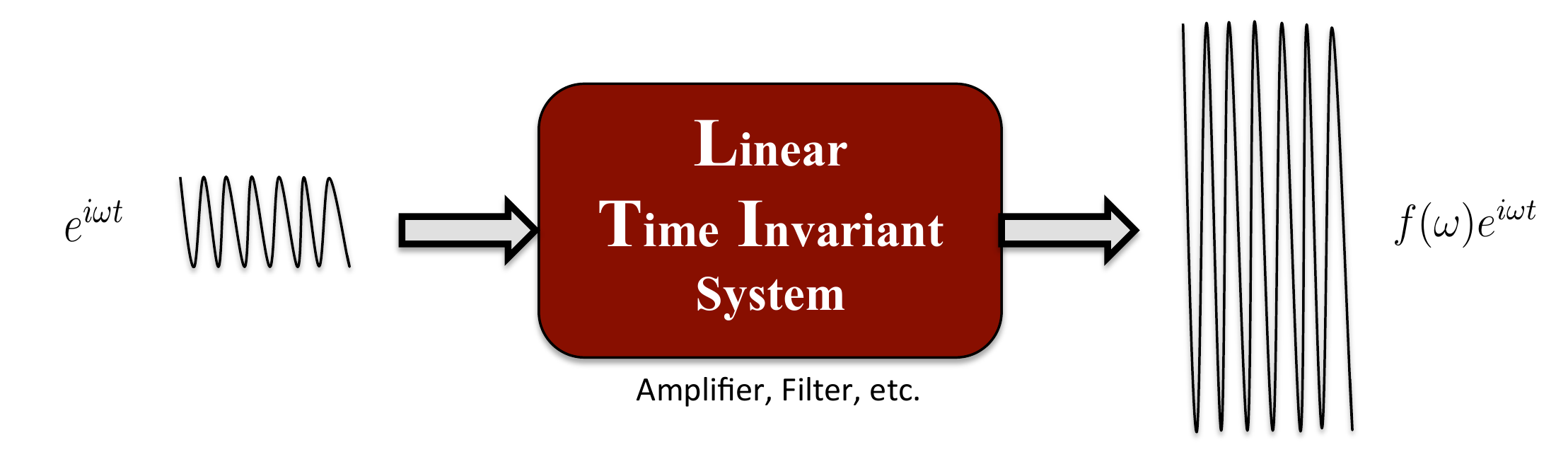}}
    \caption{\label{Fig:LTI}
  Any linear time invariant  process  transforms an input signal of frequency  $\omega$ to an output signal of the same frequency.  In other words, linearity together with time invariance implies that the process cannot change the frequency of the input. It follows that any linear time invariant system can be uniquely specified by a complex function $f(\omega)$ specifying the change in amplitude and phase of the mode $\omega$. This explains  why Fourier analysis is extremely useful for the study of these processes. 
    }
\end{figure}

 
 

We here consider an analogous decomposition of quantum states, measurements and channels into different modes.  
 The key mathematical tool is the notion of \emph{irreducible tensor operators}.  Using these, one can develop a notion of a decomposition into modes for any symmetry described by a finite or compact Lie group.
We refer to the modes appearing in such a decomposition as \emph{modes of asymmetry}. Roughly speaking, different modes of asymmetry of a state (or measurement or channel) are different characteristic ways in which it can break a given symmetry. If for a given symmetry group, a state does not have a particular mode of asymmetry then under a symmetric dynamics it can never evolve to a state which has that mode of asymmetry. Similarly, it cannot be used as a quantum reference frame for simulating measurements or channels which have that mode of asymmetry. 
 We also introduce some novel measures of asymmetry (i.e. asymmetry monotones) that can quantify the amount of asymmetry associated to a particular mode.
 
This approach provides us with a powerful tool for the study of asymmetry,
one that is particularly well adapted to understanding asymmetric quantum states of finite-dimensional systems, i.e. quantum reference frames, as physical resources. For example, these tools allow one to determine which aspects of the quantum reference frame are relevant for the degree of success that can be achieved in a reference frame alignment protocol and more generally in covariant quantum estimation problems. Similarly, they allow one to determine which aspects of the quantum reference frame state are relevant for being able to simulate asymmetric channels or asymmetric measurements.

Previous work has sometimes identified, for certain tasks such as simulating measurements and channels, which properties of a quantum reference frame are relevant for performing that task, but these insights were achieved in an ad hoc manner and only for particular groups (See e.g. \cite{BRST06,BRST06b,Poulin-Yard,Sher-Bart}).  The framework presented in this paper provides a unified and systematic way of determining what aspects of an asymmetric state are relevant for any such task, and it can also be applied to any finite or compact Lie group.

In the following we provide a couple of examples of results that one can derive with this framework.

\subsection{Some examples of applications}

\subsubsection{Spin-$j$ system as quantum reference frame}

Many authors have considered the example of
a spin-$j$ system as a quantum reference frame for direction (See e.g. \cite{BRST04, BRST06b, BRST06, BRS07, Poulin-Yard, Mar-Man,Sher-Bart, BIM,Ahm-Rud,MarvianSpekkensWAY}). In particular, one interesting question which has been studied in several works is the problem of simulating measurements and channels that break rotational symmetry using rotationally-invariant interactions and the resource of a spin-$j$ system as a quantum reference frame. 

For instance, Ref.~\cite{Sher-Bart} considers this problem for the special case of  simulating channels and measurements on a spin-$\tfrac{1}{2}$ system using a spin-$j$ system as resource. To simplify the problem, it is assumed that the state of the spin-$j$ system is invariant under rotations around a direction $\hat{n}$, which is to say that it merely acts as a reference direction rather than as a full Cartesian frame. Using this assumption, Ref.~\cite{Sher-Bart} argues that the state of the quantum reference frame can be uniquely specified by $2j$ real numbers corresponding to $2j$ moments of ${L}_{\hat{n}}$, i.e. $\{\text{tr}(\rho {L}^{k}_{\hat{n}}): 1 \le k\le 2j \}$ where $\rho$ is the state of the quantum reference frame and  where  $L_{\hat{n}}$ is the angular momentum operator in the $\hat{n}$ direction.   This characterization is then used  to quantify how well a given measurement or channel on a spin-$\tfrac{1}{2}$ system can be simulated.

As an example of the applications of our general results we reconsider this problem in section \ref{sec:Spin-j}  and we show how our approach leads to a great deal of simplification. In particular, we show that 
the quality of simulating a measurement (respectively a channel) on a spin-$\tfrac{1}{2}$ system  only depends on $\text{tr}(\rho L_{\hat{n}})$  (respectively ($\text{tr}\rho L_{\hat{n}})$ and $\text{tr}(\rho L^{2}_{\hat{n}})$). In other words, all the higher moments are irrelevant in this problem.  More generally, we consider the problem of simulating measurements and channels on a spin-$l$ system instead of a spin-$\tfrac{1}{2}$. In this case we show   that the quality of simulating a  measurement (respectively a channel)  depends only on  $2l$ (respectively $4l$)  real parameters corresponding to the moments $\{\text{tr}\left(\rho L^{k}_{\hat{n}}\right): 1 \le k\le 2l\}$ (respectively $\{\text{tr}\left(\rho L^{k}_{\hat{n}}\right): 1 \le k\le 4l\}$).

Finally, we consider the general case where the spin-$j$ system does not have any symmetries, and hence can act as a full Cartesian frame.  In theorem \ref{prop-gen-spin-l} we show that in this general case  the quality of simulating  a measurement (respectively a channel) is determined uniquely by the expectation value of all irreducible spherical tensor operators with rank less than or equal to $2l$ (respectively $4l$), i.e. by $(2l+1)^{2}-1$ (respectively $(4l+1)^{2}-1$) real parameters.

This example exhibits the power of the simple idea of mode decompositions of states, measurements and channels.

\subsubsection{Bounds on Non-deterministic Amplifiers}

An interesting application of the theory of asymmetry  is to study the quantum noise generated by optical or electronic amplifiers. Such noise is inevitable in any amplification process and has a quantum origin. The traditional explanation is based on commutation relations in together with the linearity of the equations of motion \cite{Caves}. But this approach  cannot be applied to study the noise generated in non-linear amplifiers. Also, \emph{non-deterministic amplifiers}, which were introduced in \cite{Ralph-Lund}, fall outside the scope of applicability of this  approach.  A non-deterministic amplifier is defined as one that is allowed to only succeed with some nonzero probability (it must, however, produce a flag specifying whether it has succeeded or not), and this makes it possible to achieve amplification which produces less noise than would a deterministic amplifier, when it succeeds.

It turns out that the quantum noise generated in an optical amplification process can be explained as a consequence of a $U(1)$ symmetry of the amplifier. This is the symmetry corresponding to the phase shifts of the input and output signals. So, a phase insensitive amplifier can be thought as an open-system dynamics with $U(1)$ symmetry (amplification is necessarily an open-system dynamics because it requires a source of energy). Therefore, we can apply the general theory of asymmetry to find the consequences of symmetry in this open system dynamics. In particular, we can explain the origin of quantum noise in the following way: a symmetric dynamics cannot increase the amount of asymmetry. If a phase-insensitive amplifier did not generate noise and, for example, perfectly transformed a coherent state to another coherent state with larger amplitude, then the asymmetry of the output signal would be larger than the asymmetry of the input. 
This statement can be made quantitative using the notion of asymmetry monotones (see section \ref{sec:U(1)-mon}). 

In this paper, we introduce a particular type of asymmetry monotone, one which quantifies the amount of asymmetry in a particular mode.  Under a phase-insensitive amplifier, the amount of asymmetry in each mode is non-increasing. 
Therefore, for every mode of asymmetry, we get an independent constraint on the output signal.
The advantage of this approach for explaining the origin of the noise is that it can be applied to a much broader range of amplification processes. In particular, it can be applied to nonlinear and non-deterministic amplifiers. 
In the following, we present an example of some results which can be obtained in this way.

Let $\hat{N}$ be the number operator with eigenvectors $\{|n\rangle, n=0,1,\cdots\}$ such that $\hat{N}|n\rangle=n|n\rangle$. Here, the  eigenvalue $n$ corresponds to the number of photons (excitations) in the input/output signals. This means that a phase shift $\phi$ is described by the unitary $e^{i\phi\hat{N}}$.

Consider a general input state described by the density operator $\rho=\sum_{n,m}\rho_{nm} |n\rangle\langle m|$. Suppose that under a phase-insensitive non-deterministic amplification process this state is transformed to $\rho'=\sum_{n,m}\rho'_{nm} |n\rangle\langle m|$ with probability of success $p$. Then, we show in section \ref{sec:U(1)-mon} that the following inequalities hold  
\begin{equation}
\forall k\in \mathbb{N}:\ \ \ p\le \frac{\sum_{n} |\rho_{n+k,k}|}{\sum_{n} |\rho'_{n+k,k}|}.
\end{equation}
The $k$th such inequality is a constraint derived from the non-increase of asymmetry in the $k$th mode.
In particular, these inequalities imply that if $\rho=|\psi\rangle\langle\psi|$ and $\rho'=|\psi'\rangle\langle\psi'|$ are pure states and $|\psi\rangle=\sum_{n}\psi_{n}|n\rangle$ and $|\psi'\rangle=\sum_{n}\psi'_{n}|n\rangle$   then
\begin{equation}\label{pure-state-mon}
\forall k\in \mathbb{N}:\ \ \ p\le \frac{\sum_{n} |\psi_{n+k}||\psi_{n}|}{\sum_{n} |\psi'_{n+k}||\psi'_{n}|}.
\end{equation}
To see an example of the consequences of these constraints, assume the input is a coherent state $|\alpha\rangle\equiv e^{-|\alpha|^2/2} \sum_n ({\alpha^n}/{\sqrt{n!}}) |n\rangle$. There have been speculations that  a nondeterministic quantum amplifier might be able to transform a coherent state  $|\alpha\rangle$  to a coherent state $|\alpha'\rangle$ for $|\alpha'|> |\alpha|$ with a nonzero probability~\cite{Ralph-Lund}. Then, using Eq.~(\ref{pure-state-mon}) we find
\begin{align}
\forall k\in \mathbb{N}:\ \ \ p&\le \frac{e^{-\frac{|\alpha|^{2}}{2}}\sum_{n} |\alpha|^{2n+k}/\sqrt{n!(n+k)!}}{e^{-\frac{|\alpha'|^{2}}{2}}\sum_{n} |\alpha'|^{2n+k}/{\sqrt{n!(n+k)!}}}\nonumber \\ &\le e^{-\frac{|\alpha|^{2}-|\alpha'|^{2}}{2}} \left(\frac{|\alpha|}{|\alpha'|}\right)^{k}.
\end{align}
In the limit $k\to\infty$, we can easily see that if $|\alpha'|> |\alpha|$, then the probability of transforming  the coherent state $|\alpha\rangle$ to $|\alpha'\rangle$ is zero. The conclusion that the probability of achieving a nontrivial amplification of a coherent state is strictly zero was also found in Refs.~\cite{MenziesCroke} and \cite{Amp13} by completely different arguments. 


\subsection{The structure of this paper}

We begin by explaining the main ideas of the article using the simple example of a U(1)-symmetry associated with a (non-projective) unitary representation. Then we present the generalization to the case of arbitrary finite or compact Lie group with arbitrary projective representations. In the following we present an overview of the contents of the article.

In section~\ref{sec:U(1)}, we present the idea of a mode decomposition for the special case of the group U(1). In section \ref{sec:U(1)-mon}, we introduce asymmetry monotones which quantify the amount of asymmetry in each mode.  Then in sections \ref{sec:misalign} and \ref{sec:align-q} we present some applications of the idea of mode decompositions in the context of phase references. 


To generalize the concept  of modes of asymmetry to arbitrary finite and compact Lie groups, we use the notion of  irreducible tensor operators. We provide a short review of this subject in section \ref{sec-rev}. In section \ref{sec:new-rep-g-cov},  we use this notion to  introduce a representation of \emph{$G$-covariant quantum operations}, i.e., quantum operations which have symmetry relative to a unitary representation of group $G$.  This representation of $G$-covariant operations  basically characterizes them in terms of how they act on the irreducible tensor operators. 
We use this representation to define the notion of modes of asymmetry of states for arbitrary finite and compact Lie group in section \ref{Arbit-group}, and
we introduce the notion of asymmetry monotones that quantify the amount of asymmetry in a particular mode.

In section \ref{sec:Quant-Mode}, we generalize the idea of mode decompositions to quantum channels and measurements. The main motivation for this generalization is to study the problem of simulating quantum channels and measurements using quantum reference frames. This is done in section \ref{Modes of Q.RF} where we show how the mode decomposition of a quantum reference frame determines the measurements and channels which can be simulated by it. 
Finally, in Section \ref{sec:Spin-j}, we apply these results to the important example of a spin-$j$ system as a directional quantum reference frame.

\section{Modes of asymmetry for the group U(1)} \label{sec:U(1)}

Let $e^{i\theta}\rightarrow U(\theta)$ be an arbitrary  unitary representation of the group $U(1)$. Let $\{|n,\alpha\rangle\}$ be an orthonormal basis in which the representation $e^{i\theta}\rightarrow U(\theta)$ is decomposed into irreducible representations as
\begin{equation}
U(\theta)=\sum_{n} e^{i n \theta } \sum_{\alpha} |n,\alpha\rangle\langle n,\alpha |
\end{equation}
where the integer $n$ specifies the irrep of $U(1)$ and $\alpha$ is the multiplicity index.


Let $\mathcal{B}(\mathcal{H})$ be the space of linear operators on $\mathcal{H}$, which is clearly spanned by $\{ |n,\alpha\rangle\langle m,\beta|\}$. Consider the subspace in $\mathcal{B}(\mathcal{H})$ spanned by  operators  $\{\forall n,\alpha,\beta :\   |n+k,\alpha\rangle\langle n,\beta| \}$. We denote this subspace by $\mathcal{B}^{(k)}$. We call any operator in this subspace a \emph{mode $k$ operator}.

Suppose $A^{(k)}$ is a mode $k$ operator, i.e., it lives in $\mathcal{B}^{(k)}$. We may then write it as
\begin{equation}
A^{(k)}= \sum_{n,\alpha,\beta} |n+k,\alpha\rangle\langle n,\beta|\ \textrm{tr}(A |n,\beta\rangle\langle n+k,\alpha|).
\end{equation}
It also follows that 
\begin{equation}\label{eq:modekoperator}
U(\theta) A^{(k)} U^\dag(\theta)=e^{i k\theta }A^{(k)}, \; \forall \theta \in [0,2\pi).
\end{equation}
On the other hand, if an operator $A$  satisfies Eq.~\eqref{eq:modekoperator}
 then by virtue of the linear independence of functions  $\{e^{ik\theta}\}$ we can conclude that $A$ necessarily lives in the subspace $\mathcal{B}^{(k)}$. Therefore, we have
\begin{equation}\label{transformation rule for U(1)}
A^{(k)}\in \mathcal{B}^{(k)}\ \ \   \Longleftrightarrow \ \ \forall \theta:\  U(\theta) A^{(k)} U^\dag(\theta)=e^{i k\theta }A^{(k)}
\end{equation}
For an arbitrary operator $A$, we can express it as a decomposition $A=\sum_{k} A^{(k)}$ where $A^{(k)}\in  \mathcal{B}^{(k)}$.  Here, $A^{(k)}$ is called the component of $A$ in mode $k$.
Note that for all $k\neq 0$ we have $\text{tr}(A^{(k)})=0$.  Furthermore
\begin{equation}
U(\theta) A U^\dag(\theta)=\sum_k e^{i k\theta } A^{(k)}.
\end{equation}
So to decompose  a given operator $A$ to its modes we can use the following relation
\begin{equation}
\forall k:\ \ A^{(k)}= \frac{1}{2\pi}\int d\theta\ e^{-i k\theta } \ U(\theta) A U^\dag(\theta).
\end{equation}
Note that for any  Hermitian operator $A$ it holds  that ${A^{(k)}}^\dag=A^{(-k)}$.

Suppose $\mathcal{E}$ is a U(1)-covariant super-operator, i.e.
$$\forall \theta:\ \  \mathcal{E}\left(U(\theta)(\cdot)U^\dag(\theta)\right)=U(\theta)\mathcal{E}\left(\cdot\right)U^\dag(\theta).$$
Then,  if both sides of this equation act on  an arbitrary operator $A^{(k)}\in\mathcal{B}^{(k)}$,  we get
\begin{equation} \label{U(1)-inv}
e^{ik \theta} \mathcal{E}(A^{(k)})= U(\theta) \mathcal{E}(A^{(k)})U^\dag(\theta)
\end{equation}
We can then infer from Eq.~(\ref{transformation rule for U(1)}) that $\mathcal{E}(A^{(k)})$ also lives in $\mathcal{B}^{(k)}$. 

Note that this result did not require $\mathcal{E}$ to be a completely positive map nor to be trace-preserving, but it certainly applies in these cases.  We use the term \emph{quantum operation} to refer to a completely-positive trace-nonincreasing superoperator, and \emph{quantum channel} to refer to a deterministic (i.e. trace-preserving) quantum operation. 
 We infer that U(1)-covariant quantum operations cannot change the mode of a state; they just map an operator in one mode to another operator in the same mode. In particular,
if a U(1)-covariant channel $\mathcal{E}$ maps state $\rho$ to $\sigma$ then
\begin{equation}
\mathcal{E}(\rho^{(k)})=\sigma^{(k)},
\end{equation}
where
\begin{align}
\rho&=\sum_k \rho^{(k)} \text{ with } \rho^{(k)}\in \mathcal{B}^{(k)}, \text{ and} \\
\sigma&=\sum_k \sigma^{(k)} \text{ with } \sigma^{(k)}\in \mathcal{B}^{(k)},
\end{align}
are the mode decompositions of $\rho$ and $\sigma$.

This suggests that we can interpret different $k$ as different modes of asymmetry: they cannot be interconverted to each other under U(1)-covariant quantum channels. In particular, if the initial state does not have a particular mode, then the final state of a U(1)-covariant dynamics also does not have that mode. (Of course, a mode can be eliminated if the associated component is mapped to zero by the dynamics.) Furthermore, a state $\rho$ is U(1)-invariant if and only if it contains only mode $k=0$.

Let $\text{Modes}(\rho)\equiv \{ k : \rho^{(k)}\ne 0 \}$ be the set of all integer $k$'s for which  the state $\rho$ has a nonzero component in mode $k$ (This will always include $k=0$).  
So using this notation the above observation can be summarized as follows.
\begin{proposition}\label{pro-Mode}
Assume a state $\rho$ can be transformed to another state $\sigma$ under a U(1)-covariant operation (deterministic or stochastic). Then
\begin{equation}
\text{Modes}(\sigma)\subseteq \text{Modes}(\rho).
\end{equation}
\end{proposition}

This proposition can  be understood as a refined version of the simple fact that if the initial state of a U(1)-covariant operation is invariant under a U(1)-subgroup then the final state will also be invariant under that U(1)-subgroup. To see this, first recall that under the action of the symmetry group, state $\rho$ transforms as
\begin{equation}\label{mode-decomp-rot}
U(\theta) \rho U^\dag(\theta)=\sum_k e^{i k\theta } \rho^{(k)}.
\end{equation}
Now suppose a state $\rho$ is invariant under the unitary $U(\frac{2\pi}{l})$ for some integer $l$ such that $U(\frac{2\pi}{l}) \rho U^{\dag}(\frac{2\pi}{l})=\rho$.  Using Eq.~(\ref{mode-decomp-rot}) and noting that the set $\{\rho^{(k)}\}$ are all linearly independent, we can conclude that for all modes $k$ which are not equal to an integer time $l$, $\rho^{(k)}$ must be equal to zero. On the other hand, if for all $k$'s which are not equal to some integer times $l$, $\rho^{(k)}=0$ then the state is invariant under  $U(\frac{2\pi}{l})$. So, we conclude that $\text{Modes}(\rho)$ uniquely specifies the symmetries of $\rho$, i.e., all U(1)-subgroups which leave  $\rho$ invariant.

\begin{example}
Consider a pure state  $|\psi\rangle=\sum_{n,\alpha} \psi_{n,\alpha}|n,\alpha\rangle$. Let $\Delta(\psi)$ be the difference between the highest and lowest $n$ for which  $\sum_{\alpha} |\psi_{n,\alpha}|^2$ is nonzero. Then clearly, $\Delta(\psi)=\max\{\text{Modes}(\psi)\}$. Now  proposition \ref{pro-Mode} implies that   if there exists a U(1)-covariant channel which transforms a pure state $|\psi\rangle$ to another pure  state $|\phi\rangle$ with a nonzero probability then $\text{Modes}(\phi)\subseteq \text{Modes}(\psi)$. This implies that $\max\{\text{Modes}(\phi)\} \le \max\{\text{Modes}(\psi)\}$ and therefore $\Delta(\phi)\le \Delta(\psi)$. This result has been obtained in Ref.~\cite{GS07} using a totally different argument\footnote{The proof in Ref.~\cite{GS07} proceeds by first finding a characterization of the Kraus operators of U(1)-covariant channels and then finding which pure state transformations are possible under quantum channels with this type of Kraus decomposition.}.  So the above proposition capture this result as a particular case.
\end{example}

We finish this section by providing a list of useful facts about modes of asymmetry:

\textbf{1) Modes of asymmetry of a joint system}:
Suppose $\rho_1$ and $\rho_2$ are two states with the mode decompositions
$$\rho_1=\sum_k \rho_1^{(k)}: \rho_{1}^{(k)}\in \mathcal{B}^{(k)},\ \ \ \text{and}\ \ \ \rho_2=\sum_l \rho_2^{(l)}: \rho_{2}^{(k)}\in \mathcal{B}^{(k)}$$
 We denote the mode decomposition of  $\rho_1\otimes \rho_2$ as $\rho_1\otimes \rho_2=\sum_j {(\rho_1\otimes \rho_2)}^{(j)}$. Then  we can easily see that
\begin{equation}
{(\rho_1\otimes \rho_2)}^{(j)}= \sum_k \rho_1^{(k)}\otimes \rho_2^{(j-k)}.
\end{equation}

\textbf{2) Mode decomposition for a weighted twirling operation}: Let $p(\theta)$ be an arbitrary probability density and
\begin{equation}
\sigma\equiv \int d\theta\  p(\theta)\  U(\theta) \rho U^\dag(\theta).
\end{equation}
Let
$$\rho=\sum_k \rho^{(k)}: \rho^{(k)}\in \mathcal{B}^{(k)},\ \ \ \text{and}\ \ \  \sigma=\sum_k \sigma^{(k)}: \sigma^{(k)}\in \mathcal{B}^{(k)}$$ 
be the mode decomposition of $\rho$ and $\sigma$.  Then
\begin{equation}\label{eq:modedecomp}
\sigma^{(k)}=p_{-k} \rho^{(k)},
\end{equation}
where $p_{k}=\int d\theta\  p(\theta) e^{-i\theta k}$
is the $k$th component of the Fourier transform of $p(\theta)$.

\subsection{ Quantifying the degree of U(1)-asymmetry in a given mode}\label{sec:U(1)-mon}

Asymmetry monotones are functions from states to real numbers which quantify the amount of symmetry-breaking of any given state, such that the value of these functions are non-increasing under symmetric dynamics. The intuition is that since symmetric dynamics cannot generate  asymmetry, any measure of asymmetry should be non-increasing  under this type of dynamics.  We take this as the defining property of asymmetry monotones. 
Introducing the notation $\rho\xrightarrow{\text{G-cov}}\sigma$ to denote the fact that there exists a $G$-covariant channel which transforms state $\rho$ to state $\sigma$, the definition is as follows~\cite{BRST06,GS07}. 
\begin{definition}
A function $F$ from states to real numbers is an asymmetry monotone  if $\rho\xrightarrow{\text{G-cov}}\sigma$ implies $F(\rho)\ge F(\sigma)$.  
\end{definition}
 Recently several examples of asymmetry monotones have been proposed  \cite{MarvianSpekkensNoether,Vac-Wise-Jac,GS07,Skot-Gour,Tol-Gour-Sand,GMS09}. 

In this section, we consider the problem of quantifying the amount of asymmetry \emph{in each mode}. In other words, we find asymmetry monotones which only measure the degree of asymmetry associated with some specific mode of asymmetry.

One family of such monotones can be constructed from the trace-norm. Recall that for an arbitrary operator $X$ the trace-norm of $X$ is $\|X \|\equiv \text{tr}(\sqrt{X^{\dag} X})$. This norm is  non-increasing under quantum channels (trace preserving, completely positive linear super-operators). So for any arbitrary quantum channel $\mathcal{E}$, we have $$\|\mathcal{E}(X)\| \leq \|X\|.$$

In the previous section we have seen that if $\mathcal{E}$ is a U(1)-covariant channel which maps state $\rho$ to $\sigma$ (with the mode decomposition  $\rho=\sum_k \rho^{(k)}: \rho^{(k)}\in \mathcal{B}^{(k)},\ \ \ \text{and}\ \ \  \sigma=\sum_k \sigma^{(k)}: \sigma^{(k)}\in \mathcal{B}^{(k)}$) then
$\forall k:\ \mathcal{E}(\rho^{(k)})=\sigma^{(k)}$. Now the monotonicity of the trace-norm implies
$$\forall k:\  \|\sigma^{(k)} \| \le \|\rho^{(k)} \|.$$
So we can think of  $\|\rho^{(k)} \|$  as a measure of the amount of asymmetry of the state $\rho$ in the mode $k$. 

Now suppose a given state $\rho$ can be transformed to another state $\sigma$ under a U(1)-covariant channel with probability $p$. If this is possible then there exists a U(1)-covariant channel which maps state $\rho$ to
$$\tilde{\sigma}\equiv p\ \sigma\otimes |\text{succ}\rangle\langle \text{succ} |+ (1-p) \frac{I}{d}\otimes |\text{fail}\rangle\langle \text{fail} |, $$
where $ |\text{succ}\rangle, |\text{fail}\rangle$ are two orthonormal states which are invariant under the symmetry transformations and   $\frac{I}{d}$ is the completely mixed state on the Hilbert space of $\sigma$ and is  clearly invariant under all symmetry transformations.  Now the fact that this channel is U(1)-covariant implies that for all $k$:
$|\tilde{\sigma}^{(k)} | \le |\rho^{(k)} |$.
However, because states $ |\text{succ}\rangle, |\text{fail}\rangle$  and   $\frac{I}{d}$ are invariant under the symmetry transformations this implies that for all $k\ne 0 $  it holds that
$$
\|{\tilde{\sigma}}^{(k)} \|=\left\|\frac{1}{2\pi}\int d\theta e^{-i\theta k} U(\theta) \tilde{\sigma} U^{\dag}(\theta)\right\|= p \|{\sigma}^{(k)} \| \le \|\rho^{(k)} \|.
$$
So to summarize, we have shown that
\begin{proposition}\label{prop-mode-mon-U(1)}
Suppose there is a U(1)-covariant channel which maps a state $\rho$ to state $\sigma$ with probability $p$. Then  it holds that
\begin{equation}\label{eq:prop4}
\forall k:\ \ p {\|\sigma^{(k)} \|} \le { \|{\rho}^{(k)} \| }.
\end{equation}
\end{proposition}
This proposition can be thought of as a quantitative version of proposition \ref{pro-Mode}.

Using a similar argument, one can prove the following more general proposition about transforming a state to an ensemble of states.
\begin{proposition}\label{prop-mode-mon-U(1)-ens}
Suppose there is a U(1)-covariant channel that maps the state $\rho$ to the ensemble consisting of states $\sigma_{i}$ with probabilities $p_{i}$,  where the value of $i$ becomes known at the end of the process. Then  it holds that
\begin{equation}\label{eq:prop5}
\forall k:\ \ \sum_{i} p_{i} {\|\sigma_{i}^{(k)} \|} \le { \|{\rho}^{(k)} \| }.
\end{equation}
\end{proposition}
This result subsumes proposition~\ref{prop-mode-mon-U(1)} as a special case because Eq.~\eqref{eq:prop5} implies that for any given value of $i$, $\forall k:\ \ p_{i} {\|\sigma_{i}^{(k)} \|} \le { \|{\rho}^{(k)} \| }.$

In the following, we calculate $\|\rho^{(k)} \|$ for arbitrary state $\rho$ in the case where the representation is multiplicity free, so that $U(\theta)=\sum_{n} e^{i\theta n} |n\rangle\langle n|$. (Note that all the previous results work for any representation of $U(1)$ no matter if the representation has multiplicity or not.)  Consider an arbitrary density operator $\rho=\sum_{n,m} \rho_{n,m} |n\rangle\langle m|$.  Then
$$\rho^{(k)}=\sum_{n} \rho_{n+k,n}|n+k\rangle\langle n|.$$
Therefore
\begin{equation}
\|\rho^{(k)} \|=\text{tr}\left(\sqrt{\rho^{(k)}\rho^{(k)}{ }^{\dag} } \right)=\sum_{n} |\rho_{n+k,k}|.
\end{equation}
In particular, if the state is pure, i.e., $\rho=|\psi\rangle\langle\psi |$ where  $|\psi\rangle=\sum_{n}\psi_{n}|n\rangle$, then
\begin{equation}
\|\rho^{(k)} \|=\sum_{n} |\psi_{n+k}||\psi_{n}|.
\end{equation}
Also, note that if the state is pure, then
\begin{equation}\label{eq:maxvalueasymmetry}
\|\rho^{(k)} \| \le 1,
\end{equation}
where the bound follows from the Cauchy-Schwartz  inequality.

It is also worth noting that the sum of this monotone over all modes, $\sum_k \|\rho^{(k)} \|= \sum_{k,n} |\rho_{n+k,k}|$, is also an asymmetry monotone.  It is equivalent to the asymmetry monotone presented in Eq.~(44) of Ref.~\cite{TolouiGour}, with the equivalence manifest when the expression is worked out for pure states, in Eq.~(49).

\begin{example}
Consider the sequence of states
\begin{equation}
\left\{ |\psi_{N}\rangle\equiv\frac{1}{\sqrt{N }} \sum_{n=1}^{N} |n\rangle : N\in \mathbb{N} \right\}.
\end{equation}

One can easily see that for any given state $|\phi\rangle$ there is a U(1)-covariant channel $\mathcal{E}_{N}$ which transforms $|\psi_{N}\rangle$ to a state arbitrary close to $|\phi\rangle$ in the limit of large $N$.   $\mathcal{E}_{N}$  is given by
\begin{equation}
\mathcal{E}_{N}(\rho)= \int d\theta\ U({\theta})|\phi\rangle\langle\psi_{N}| U^{\dag}(\theta)\ \rho \ U(\theta) |\psi_{N}\rangle\langle\phi|U^{\dag}(\theta).
\end{equation}
The sufficiency of $|\psi_N\rangle$ for forming approximations to any other state in the limit of large $N$ suggests that $|\psi_{N}\rangle$ has the maximal possible asymmetry in this limit.  This can be made precise as follows.
\begin{align*}
\left\||\psi_{N}\rangle\langle\psi_{N}|^{(k)} \right\|=\sum_{n} |\psi_{n+k}||\psi_{n}|&=1-\frac{|k|}{N}\ \ \ |k|\le N \\  &=0\ \ \ \ \ \ \  \text{otherwise.}
\end{align*}
So for all modes $k$ for which $|k|\ll N$, the  state  $|\psi_{N}\rangle$ has almost the maximal value of asymmetry for mode $k$ with respect to this monotone (namely, the value 1, as shown in Eq.~\eqref{eq:maxvalueasymmetry}). 
\end{example}

\subsection{Effect of misalignment of phase references} \label{sec:misalign}

To be able to measure a quantity with high precision one fundamental requirement is to have a precise reference frame, for instance, in the case of measuring a time interval, a high precision clock.  Any uncertainty in the configuration of the reference frame will limit the precision of the measurements that one can perform.

In this section, we consider the problem of misalignment of  phase references. So we assume the system under consideration carries a non-trivial representation of the group U(1) given by $e^{i\theta}\rightarrow U(\theta)$. The U(1) group may have different physical interpretations: It may describe a rotation around some axis or a phase shift  between states with different numbers of photons.



We assume there is an ideal perfect reference frame possessed by Alice and there is a noisy reference frame possessed by Bob. For example, Bob can be on a satellite  and so has access to a clock with low accuracy while Alice is on earth and has access to a high precision atomic clock.

 Assume they know that the phase shift relating Bob's reference frame to Alice's is $\theta$ with probability $p(\theta)$.
If $\theta$ were known, then a state which is described by $\rho$ relative to Alice's reference frame would be described by $U(\theta)\rho U^{\dag}(\theta)$ relative to Bob's reference frame.  Given that $\theta$ is only known to be distributed according to $p(\theta)$, it follows that the state is described relative to Bob's reference frame as 
\begin{equation}
\widetilde{\rho}\equiv \int d\theta\ p(\theta)\ U(\theta)\rho U^\dag(\theta),
\end{equation}
which is generally a mixed state.  This explains how the lack of a perfect reference frame can limit Bob's ability to get information about an unknown state $\rho$.


 By Eq.~\eqref{eq:modedecomp}, the state $\tilde{\rho}$ can be rewritten in terms of the mode decomposition of  $\rho$ as
\begin{equation}
\widetilde{\rho}=\sum_k p_{-k} \rho^{(k)},
\end{equation}
where $\rho^{(k)}$ is the $k$th component of the mode decomposition of $\rho$ and
\begin{equation}
p_{k}=\int d\theta\ p(\theta) e^{-i\theta k}
\end{equation}
is the $k$th component of the Fourier transform of $p(\theta)$.



So to understand how the uncertainty about the phase reference can affect Bob, it is helpful to consider the Fourier transform of the probability distribution $p(\theta)$. For example, if the Hilbert space of the system under consideration carries a finite number of irreps of U(1), then there will be a finite set of modes in which a state can have nonzero components. Then any quantity which quantifies the effect of misalignment described by the probability distribution $p(\theta)$ should only depend on the Fourier component of $p(\theta)$ in those particular modes. 
\subsubsection{Example}

Consider the representation of U(1) given by $e^{i\theta}\rightarrow U(\theta)$ where
$$U(\theta)=\sum_{n=n_{\text{min}}}^{n_{\text{max}}} e^{i\theta n} |n\rangle\langle n|$$ Assume the phase difference between Alice and Bob's local reference frames is $\theta$ with probability $p(\theta)$. Now, to quantify the effect of this misalignment on the description of an arbitrary state of this system we only need to consider the Fourier components of $p(\theta)$, denoted $p_{k}$, for $0 \le k \le n_{\text{max}}-n_{\text{min}}$. For example, suppose Bob wants  to estimate the phase $\phi$ of the state
\begin{equation}
\frac{1}{\sqrt{2}}\left(|n=0\rangle+e^{i\phi}|n=l\rangle\right).
\end{equation}
The information about this phase lives only in the modes $l,-l$. So the only property of the probability distribution $p(\theta)$ which is relevant for this estimation problem is its $l$th Fourier component.
In particular,  if this component
is zero, then Bob cannot get any information about the phase $\phi$. This can happen even if the two phase references are highly correlated. For example, if
 $$p(\theta)=\frac{1}{2}\delta(\theta)+\frac{1}{4}\delta(\theta-\pi/l)+\frac{1}{4}\delta(\theta+\pi/l),$$
 then $p(\theta)$ has no component in the mode $l$ and therefore Bob cannot get any information about the phase $\phi$. This simple observation shows that a measure of the alignment of two reference frames should be chosen based on the specific task to which the reference frames are being applied.

 In many practical situations we can assume that the probability distribution $p(\theta)$ is almost Gaussian. In particular, this is the case if Bob's knowledge of $\theta$ is obtained by averaging over many independent estimations. Let $\delta \theta$ be the standard deviation of $\theta$. Then, for Gaussian distributions we know that for all modes with $|k|\ll 1/\delta \theta$, $|p_{k}|\approx 1$ and therefore for these modes the distribution is effectively a delta function over $\theta$.  So in the case of the  above example where Bob is interested in estimating the phase $\phi$ of the state $1/\sqrt{2}\left(|n=0\rangle+e^{i\phi}|n=l\rangle\right)$, if $\delta\theta \ll 1/l$ then the imperfectness of Bob's local frame does not put any significant limitation on his performance.

\subsection{Alignment of phase references using U(1)-asymmetric states}\label{sec:align-q}

If Alice wishes to ensure that Bob's reference frame is aligned with her own, she can send him a \emph{quantum reference frame}, i.e., a quantum system prepared in an asymmetric state which carries information about her reference frame. For example, Alice can send Bob many copies of the state described by $(1/\sqrt{2})(|0\rangle+|1\rangle)$ relative to her reference frame  and also tell him the description of this state relative to her reference frame. Then Bob can use these quantum systems to obtain information about the relative phase between his reference frame and Alice's.

Assume Alice and Bob's prior knowledge about the phase difference between their local phase references is described by the probability distribution $p(\theta)$. Consider an arbitrary state described by  $\rho$ relative to Alice's reference frame.  As we have seen before, the lack of information regarding the relation of Bob's reference frame to Alice's prevents him from 
obtaining as much information about the unknown state $\rho$ as Alice could.
Now assume that Alice also sends Bob a quantum reference frame in the state $\tau$ and assume that the representation of phase shifts on this system is given by $e^{i\theta}\rightarrow V(\theta)$.

To find more precise information about $\rho$, Bob can first use the quantum reference frame $\tau$ to align his reference frame with Alice's and then perform some measurement on $\rho$. But, this procedure does  not describe the most general process that Bob can implement. The most general process is to perform a joint measurement on the state $\rho$ and the quantum reference frame $\tau$. In this case the information Bob can obtain about the unknown state $\rho$ is the information he can extract from the state
\begin{equation}
\int d\theta \ p(\theta) \left(V(\theta)\otimes U(\theta)\right)   \tau\otimes\rho  \left(V^{\dag}(\theta)\otimes U^{\dag}(\theta)\right).
\end{equation}
This state is equal to
\begin{equation}
 \sum_{k_{1},k_{2}} p_{-k_{1}} \tau^{(-k_{2})} \otimes \rho^{(k_{1}+k_{2})},
\end{equation}
where $\tau=\sum_{k}\tau^{(k)}$ is the mode decomposition of $\tau$ and $p_{k}$ is the $k$th Fourier component of $p(\theta)$. This shows precisely how the information Bob can obtain about
different modes of $\rho$ is determined by which modes are present in the state of the quantum reference frame and in the probability distribution $p(\theta)$ decribing the misalignment.

\subsubsection{Example}

Suppose Alice and Bob's local reference frames are initially  uncorrelated and therefore the prior distribution $p(\theta)$ is uniform.

Assume Bob wants to find information about the phase $\phi$ of the state
\begin{equation} \label{eq:k2state}
\frac{1}{\sqrt{2}}\left(|n=0\rangle+e^{i\phi}|n=2\rangle\right)
\end{equation}
Note that here the information is encoded in the modes $2$ and $-2$. So to enable Bob to encode this information Alice should send him a quantum reference frame which has modes  $2$ and $-2$. In particular, the reference frame should not be invariant under $U(\pi)$, because if $U(\pi)\tau U^{\dag}(\pi)=\tau$ then the state $\tau$ will not have any component in mode 2. But,  lack of  this  symmetry does not imply that the quantum reference frame has mode 2. For example, assume Alice sends Bob the quantum reference frame
\begin{equation}
|{\psi}\rangle= \frac{|{0}\rangle+|{1}\rangle}{\sqrt{2}}.
\end{equation}
This state is not invariant under any subgroup of \text{U(1)}. But it still does not have any component in the mode $k=2$ and so it does not help Bob to obtain information about the phase $\phi$ of the state \eqref{eq:k2state}.

\section{Representation of G-covariant channels in the irreducible tensor operator basis} \label{sec:G-cov-chan}

In this section  we first present a short review of irreducible tensor operators (See e.g.  \cite{Cornwell} and \cite{Sakuraii} for more information on this subject.).
Then,  we  introduce a new representation of G-covariant channels which basically describes a G-covariant channel by specifying how it acts on an \emph{irreducible tensor operator basis}.

\subsection{Review of irreducible tensor operators} \label{sec-rev}


Let $\mathcal{B}(\mathcal{H})$ be  the space of all bounded operators acting on the Hilbert space $\mathcal{H}$. For any unitary $V\in \mathcal{B}(\mathcal{H})$ the super-operator $V(\cdot)V^{\dag}$ preserves the Hilbert-Schmidt inner product on $\mathcal{B}(\mathcal{H})$, defined as $\langle A,B\rangle\equiv\text{tr}(A^\dag B)$  for arbitrary $A,B\in\mathcal{B}(\mathcal{H})$. So the super-operator $V(\cdot)V^{\dag}$  can be thought of as a unitary acting on the space $\mathcal{B}(\mathcal{H})$. 

Suppose $g\rightarrow U(g)$ is a projective unitary representation of a finite or compact Lie group $G$ on the Hilbert space $\mathcal{H}$.   Then $g\rightarrow \mathcal{U}_g$ where $\mathcal{U}_{g}[\cdot]\equiv U(g)(\cdot)U^{\dag}(g)$ is a unitary representation of $G$ on $\mathcal{B}(\mathcal{H})$. Note that this representation is always non-projective,
\begin{equation}
\forall g_{1},g_{2}\in G:\ \ \mathcal{U}_{g_{2}}\circ \mathcal{U}_{g_{1}}=\mathcal{U}_{g_{2}g_{1}}.
\end{equation}

Let $\{T_{m}^{(\mu,\alpha)}\}$ be a basis of $\mathcal{B}(\mathcal{H})$ in which the representation $g\rightarrow \mathcal{U}_g$ decomposes to the irreps of $G$  such that
\begin{equation}\label{eq-def-tens-irrep}
\mathcal{U}_{g}[T_{m}^{(\mu,\alpha)}]=\sum_{m'} u^{(\mu)}_{m'm}(g) \ T_{m'}^{(\mu,\alpha)},
\end{equation}
where
\begin{equation}
u^{(\mu)}_{m'm}(g) \equiv \langle \mu,m' | U^{(\mu)}(g) |\mu,m\rangle,
\end{equation}
are the matrix elements of $U^{(\mu)}(g)$, the unitary (non-projective) irreducible representation of $G$ labeled by $\mu$.
We choose this basis to be normalized such that
\begin{equation}
\text{tr}( T_{m}^{(\mu,\alpha)}{ }^{\dag} T_{m'}^{(\mu',\alpha')}) =\delta_{\mu,\mu'}\delta_{\alpha,\alpha'}\delta_{m,m'}
\end{equation}

 Here, $\alpha$ can be thought of as a multiplicity index. We call the basis $\{T_{m}^{(\mu,\alpha)}\}$ the \emph{irreducible tensor operator basis}. Also, the elements of the set $\{T_{m}^{(\mu,\alpha)}\}$ for a fixed $\mu$ and $\alpha$ are called \emph{components} of the irreducible tensor $T^{(\mu,\alpha)}$. We call the irrep label $\mu$ the \emph{rank} of the tensor operator $T_{m}^{(\mu,\alpha)}$. 
 


 Consider the   Hermitian conjugate of both sides of Eq.~(\ref{eq-def-tens-irrep}),
 \begin{equation}\label{comp-conj}
\left(\mathcal{U}_{g}[T_{m}^{(\mu,\alpha)}]\right)^{\dag}=\mathcal{U}_{g}[T_{m}^{(\mu,\alpha)}{ }^{\dag}]=\sum_{m'} \bar{u}^{(\mu)}_{m'm}(g) \ T_{m'}^{(\mu,\alpha)}{ }^{\dag}
\end{equation}
where $\bar{u}^{(\mu)}_{m'm}(g)$ denotes the complex conjugate of ${u}^{(\mu)}_{m'm}(g)$. This implies that  for any component $T_{m}^{(\mu,\alpha)}$ of a tensor operator of rank $\mu$, its Hermitian conjugate $T_{m}^{(\mu,\alpha)}{ }^{\dag}$ is in the subspace spanned by rank $\bar{\mu}$ irreducible tensor operators where  $\bar{\mu}$ denotes the irrep equivalent to the complex conjugate of  irrep $\mu$. In particular, in the case of SO(3) (or equivalently SU(2)) where the complex conjugate of any irrep $\mu$ is equivalent to the irrep $\mu$,  the Hermitian conjugate of a component of an irreducible tensor operator with rank $\mu$ is in the subspace spanned by the irreducible tensor operators with rank $\mu$.


To find an irreducible tensor operator basis in $\mathcal{B}(\mathcal{H})$ it is helpful to use the Liouville representation of operators in which an operator will be represented by a vector formed by stacking all the rows of its matrix representation (in some specific basis defining the representation) in a column vector \cite{Sher-Bart}.  This is equivalent to the Choi isomorphism between operators on $\mathcal{H}$ and vectors on $\mathcal{H}\otimes \mathcal{H}$.

Then the Liouville (or Choi) representation of the super-operator $\mathcal{U}_{g}$ will be $U(g)\otimes \bar{U}(g)$ where  $\bar{U}(g)$ denotes the complex conjugate of $U(g)$ in the basis that defines the representation. So the ranks of all tensor operators which show up in the space  $\mathcal{B}(\mathcal{H})$ corresponds to the set of all irreps of $G$ which show up in the representation $g\rightarrow U(g)\otimes \bar{U}(g)$.  Furthermore, to decompose a particular operator in $\mathcal{B}(\mathcal{H})$ to irreducible tensor operators we can write the Liuoville representation of that operator and find out how it decomposes into the irreducible basis of $G$ defined by the representation $g\rightarrow U(g)\otimes \bar{U}(g)$.


One can construct higher ranks of irreducible tensor operators by decomposing the product of irreducible tensor operators with lower ranks. Let $\{T_{m}^{(\mu_{1})}\}$ be the components of a rank $\mu_{1}$ tensor operator and   $\{R_{m}^{(\mu_{2})}\}$  be the components of a rank $\mu_{2}$ tensor operator. Finally, let $C^{\mu_{3},m_{3},\alpha}_{\mu_{1},m_{1};\mu_{2},m_{2}}$ be the Clebsch-Gordon coefficients (see e.g. \cite{Sakuraii}).
Then the set of operators $\{S_{m}^{(\mu_{3},\alpha)}\}$ defined by
\begin{equation}\label{Clebsh-Gordon-tens}
S_{m}^{(\mu_3,{\alpha})}= \sum_{m_{1},m_{2},\mu_{1},\mu_{2} } C^{\mu_{3},m_{3},\alpha}_{\mu_{1},m_{1};\mu_{2},m_{2}}
   T_{m_{1}}^{(\mu_{1})} R_{m_{2}}^{(\mu_{2})}
\end{equation}
are components of a rank $\mu_{3}$ irreducible tensor operator.

Finally, we present the Wigner-Eckart theorem which gives a useful tool to find the irreducible tensor operator basis (See e.g. \cite{Cornwell}):

\begin{theorem}(\textbf{Wigner-Eckart})
Let $G$ be a finite group or a compact Lie group. Let $T_{m_{1}}^{(\mu_{1},\alpha)}$ be an element of a tensor operator. Then
\begin{align}
\langle\mu_{3},m_{3}&| T_{m_{1}}^{(\mu_{1},\alpha)}|\mu_{2},m_{2}\rangle \nonumber \\
&=\sum_{\beta} \left(C^{\mu_{3},m_{3},\beta}_{\mu_{1},m_{1};\mu_{2},m_{2}}\right)^{\ast} \left(\mu_{3}|T^{(\mu_{1},\alpha)}|\mu_{2}\right)_{\beta}
\end{align}
where $\beta$ is a multiplicity index that counts the number of copies of the $\mu_{3}$ irrep that can be formed by composing irreps $\mu_{1}$ and $\mu_{2}$, $C^{\mu_{3},m_{3},\beta}_{\mu_{1},m_{1};\mu_{2},m_{2}}$  are the Clebsch-Gordon coefficients for this composition and  $\left(\mu_{3}|T^{(\mu_{1},\alpha)}|\mu_{2}\right)_{\beta}$ is a number which is independent of $m_{1},m_{2}$, and $m_{3}$.
\end{theorem}
Note that the left hand side of the equality can be interpreted as the matrix elements of the unitary acting on $\mathcal{B}(\mathcal{H})$ which transforms the orthonormal basis $\{|\mu_{3},m_{3}\rangle\langle\mu_{2},m_{2} |\}$ to the orthonormal basis $\{T_{m_{1}}^{(\mu_{1},\alpha)}\}$.

\subsubsection*{Example: SO(3)}\label{Ex-tensors-SU(2)}
In the case of SO(3), the complex conjugate of any representation is unitarily equivalent  to the original representation:   Suppose $\bar{U}(g)$ is the complex conjugate of $U(g)$ in the basis in which $L_{z}$ is diagonal and all the matrix elements of $L_{x}$ are real numbers. Then
\begin{equation}
\forall g\in \text{SO(3)}\ \ \bar{U}(g)=e^{-i\pi L_{y}} U(g) e^{i\pi L_{y}}.
\end{equation}
Let $g\rightarrow U(g)$ be an arbitrary projective unitary  representation of SO(3) on $\mathcal{H}$. 
The above discussion implies that one way to find the ranks of tensor operators and their multiplicities for the basis $\{T_{m}^{(\mu,\alpha)}\}$ which spans $\mathcal{B}(\mathcal{H})$ is to find  the irreps and their multiplicities  which show up in the representation 
$$g\rightarrow U(g)\otimes \bar{U}(g).$$

An important special case, which we use later, is when $\mathcal{H}$ carries  a spin-$j$ irrep  of SO(3). Then the above observation implies that  $\mathcal{B}(\mathcal{H})$ is spanned by 
$$\{T_{m}^{(\mu)}:  (\mu,m): 0\le \mu\le 2j, -\mu\le m\le \mu \}$$ 
and there is no multiplicity.  In other words, the maximum rank of the irreducible tensor operators on this space is $2j$. 

Note that the operators $\{T_{m}^{(\mu)} \}$ are uniquely defined only when we fix the basis we use to  represent the matrix elements $u_{m'm}^{(\mu)}(g)$  in Eq.~(\ref{eq-def-tens-irrep}). In the case of SO(3), we always use the basis in which the matrix representation of $L_{z}$ is diagonal and the matrix elements of $L_{x}$ are all real numbers. 

Then, it follows that in this basis
\begin{align*}
\mu=0&:\ \ &&T^{(\mu=0)}=c_{0}\mathbb{I}\\
\mu=1&:\ \ &&T^{(\mu=1)}_{m=0}=c_1 L_{z},\ \  T^{(\mu=1)}_{m=\pm 1}=\pm c_1\frac{ 1}{\sqrt{2}}L_{\pm}
\end{align*}
where $\mathbb{I}$ is the identity operator on $\mathcal{H}$, $L_{\pm}\equiv L_{x} \pm i L_{y}$ and $c_0,c_{1}$ are normalization factors \cite{Sakuraii}.

One can generate all higher rank tensor operators on this space, by considering the products of  $T^{(\mu=1)}_{m_{1}}T^{(\mu=1)}_{m_{2}}\cdots$  and decomposing them to irreducible tensor operators using Eq.~(\ref{Clebsh-Gordon-tens}). Following this method one can show that the rank-2 irreducible tensor operators are
\begin{align*}
\mu=2: \ \ \ \  &T^{(\mu=2)}_{m=\pm 2}=c_2\frac{1}{{2}}L_{\pm}^{2},\ \ \\ &T^{(\mu=2)}_{m=\pm 1}=c_2\frac{\pm 1}{{2}}(L_{\pm} L_{z}+ L_{z} L_{\pm}),\ \ \\  &T^{(\mu=2)}_{m=0}=c_2\frac{1}{\sqrt{6}}\left(3L_{z}^{2}-L^{2}\right)
 \end{align*}
where $L^{2}=L_{x}^{2}+L_{y}^{2}+L_{z}^{2}$ is the total angular momentum and $c_{2}$ is a normalization factor (see e.g. \cite{Sakuraii}).






\subsection{A representation of G-covariant super-operators }\label{sec:new-rep-g-cov}

In this section, we introduce a representation of G-covariant super-operators which will be useful in the rest of this paper. 

Recall that a super-operator $\mathcal{E}$ is G-covariant if it commutes with the super-operator representation of the group G,
\begin{equation}
\forall g\in G: \ \  \mathcal{E}\circ \mathcal{U}_g=  \mathcal{U}_g \circ \mathcal{E}
\end{equation}
Then Schur's lemma implies that $\mathcal{E}$ should be block diagonal in any basis of the operator space $\mathcal{B}(\mathcal{H})$ which decomposes the representation $g\rightarrow\mathcal{U}_g$ into the irreps of $G$.  But, this is exactly the definition of an irreducible tensor operator basis and therefore $G$-covariant channels are block diagonal in the irreducible tensor operator bases.  The following lemma states this result.

\begin{lemma}\label{G-cov-Sup-tensor}
Let $g\rightarrow U_{\text{in}}(g)$ and $g\rightarrow U_{\text{out}}(g)$
be projective unitary representations of the group $G$ on the Hilbert spaces $\mathcal{H}_{\text{in}}$ and  $\mathcal{H}_{\text{out}}$.  Let $\{T_{m}^{(\mu,\alpha)}\}$ and $\{S_{m}^{(\mu,\beta)}\}$ be the corresponding normalized irreducible tensor operator bases for $\mathcal{B}(\mathcal{H}_{in})$ and $\mathcal{B}(\mathcal{H}_{out})$. Consider a linear superoperator $\mathcal{E}:\mathcal{B}(\mathcal{H}_{in})\rightarrow \mathcal{B}(\mathcal{H}_{out})$  which is G-covariant, i.e., $\forall g\in G:\ \mathcal{E}\left(U_{\text{in}}(g)\cdot U^{\dag}_{\text{in}}(g)\right)=U_{\text{out}}(g)\mathcal{E}\left(\cdot \right)U^{\dag}_{\text{out}}(g)$. Then \begin{equation}
\mathcal{E}(X)= \sum_{\mu,m,\alpha}  \text{tr}\left( T_{m}^{(\mu,\alpha)}{ }^{\dag} X \right)\left[ \sum_{\beta}c_{\beta\alpha}^{(\mu)}\  S_{m}^{(\mu,\beta)} \right]
\end{equation}
where $c_{\beta\alpha}^{(\mu)}\equiv \text{tr}\left( S_{m}^{(\mu,\beta)}{}^{\dag} \mathcal{E}(T_{m}^{(\mu,\alpha)}) \right)$ (which turns out to be independent of $m$).
\end{lemma}
The proof is straightforward and is presented in appendix \ref{app:proofs}.
This representation simply means that under G-covariant super-operators, an  input operator in the mode $(\mu,m)$ can only be mapped to an output operator in the same mode (for a general linear super-operator there is no such constraint on the output).

Lemma \ref{G-cov-Sup-tensor} implies that any  linear G-covariant super-operator can be uniquely specified by specifying the set of matrices $\{c^{(\mu)}\}$ for the set of all  $\mu$ which show up as ranks of irreducible tensor operators in both input and output spaces.  In the next chapter we use this representation of G-covariant super-operators to study the asymmetry properties of quantum states. It can also have  applications in other fields such as tomography of G-covariant channels or equivalently tomography of the symmetrized version of a channel (see \cite{Emerson1,Emerson2}).

\begin{example}\label{Ex-Rot-cov}
Consider a rotationally covariant  super-operator from $\mathcal{B}(\mathcal{H}_{j_{1}})$ to $\mathcal{B}(\mathcal{H}_{j_{2}})$ where  the input and output spaces $\mathcal{H}_{j_{1}}$ and $\mathcal{H}_{j_{2}}$ are spin-$j_{1}$ and spin-$j_{2}$  irreps of SO(3) respectively. 

Then, from section \ref{Ex-tensors-SU(2)} we know that the tensor operators for both input and output spaces do not have multiplicity and  their rank varies between $\mu_1^{\text{min}}=0$ and $\mu_1^{\text{max}} = 2j_{1}$ in the input space and between $\mu_2^{\text{min}}=0$ and $\mu_2^{\text{max}}=2j_{2}$ in the output space.  So lemma  \ref{G-cov-Sup-tensor}  implies that an arbitrary rotationally covariant super-operator  from $\mathcal{B}(\mathcal{H}_{j_{1}})$ to $\mathcal{B}(\mathcal{H}_{j_{2}})$  can be described by coefficients $c^{(\mu)}$ where $\mu$ varies between $\mu^{(\text{min})} = 0$  and $\mu^{(\text{max})} =\min\{\mu_1^{\text{max}},\mu_2^{\text{max}}\}.$

If this super-operator is a channel, i.e., it is trace-preserving and completely positive, then we can put more constraints on the coefficients $c^{(\mu)}$. First, we use the fact that any completely positive super-operator maps Hermitian operators to Hermitian operators. This implies that all the coefficients $\{c^{(\mu)}\}$ should be real.  On the other hand, the fact that a quantum channel is trace-preserving fixes one coefficient, namely, $c^{(\mu=0)}$. So any SO(3) covariant channel on these spaces can be described by 
$$2\min\{j_{1},j_{2}\}$$
real numbers. The special case of this result for $j_{1}=j_{2}$ has been observed previously in  \cite{Sher-Bart}. 

 In particular, if the input space is a spin-1/2 system, the channel can be described by just one real parameter. Note that in the absence of symmetry the number of parameters one needs to specify the channel scales as $j_{1}^{2}j_{2}^{2}$.

Let $\{T_{m}^{(\mu)}\}$ and $\{S_{m}^{(\mu)}\}$ be the irreducible tensor operator basis for  $\mathcal{B}(\mathcal{H}_{j_{1}})$ and $\mathcal{B}(\mathcal{H}_{j_{2}})$ and $\{c^{(\mu)}:\ \mu=1,\dots, 2\min\{j_{1},j_{2}\} \}$ be the coefficients describing the rotationally invariant super-operator $\mathcal{E}$ from  $\mathcal{B}(\mathcal{H}_{j_{1}})$ to $\mathcal{B}(\mathcal{H}_{j_{2}})$.  It follows from lemma \ref{G-cov-Sup-tensor} that if $\mathcal{E}(\rho)=\sigma$   then
\begin{equation}
\text{tr}\left(\sigma S_{m}^{(\mu)}{ }^{\dag}\right)=c^{(\mu)}\text{tr}\left(\rho  T_{m}^{(\mu)}{ }^{\dag}  \right).
\end{equation}
Finally,  recall that  the trace norm  is non-increasing under positive and trace-preserving super-operators. This implies that if the super-operator  $\mathcal{E}$ is positive and trace-preserving then $\forall (\mu,m): \|\mathcal{E}(T_{m}^{(\mu)})\|\le  \|T_{m}^{(\mu)} \|$ which by virtue of  lemma \ref{G-cov-Sup-tensor}  implies 
\begin{equation}
\forall (\mu,m): \left|c^{(\mu)} \right| \le \frac{  \|T_{m}^{(\mu)} \|}{  \|S_{m}^{(\mu)} \|}.
\end{equation}
In particular, if the input and output spaces are the same, i.e., $j_{1}=j_{2}$, then 
\begin{equation}
\forall (\mu,m): \left|c^{(\mu)}\right|\le 1.
\end{equation}
\end{example}

Consider the case where the output space of the G-covariant super-operator $\mathcal{E}_{1}$ matches the input space of $\mathcal{E}_{2}$ such that the composition $\mathcal{E}_{2}\circ\mathcal{E}_{1}$ is well-defined. If $\mathcal{E}_{1}$  is described by the set of matrices $\{c^{(\mu)}\}$  and $\mathcal{E}_{2}$  is described by the set of matrices $\{d^{(\mu)}\}$ then $\mathcal{E}_{2}\circ\mathcal{E}_{1}$ is described by the set of matrices  $\{d^{(\mu)}c^{(\mu)}\}$.
This implies that in cases such as the example above,  where all tensor operators are multiplicity free  and $c^{\mu}$ and $d^{\mu}$ are scalars, then all $G$-covariant super-operators commute with each other.
Furthermore, this observation implies that a master equation which describes a G-covariant dynamics can be decomposed to a set of uncoupled  differential equations for each of these matrices.


\section{Modes of asymmetry for an arbitrary group} \label{Arbit-group}

With the framework of irreducible tensor operators in hand, we can now generalize the notion of modes of asymmetry, which we have thus far only defined for the case of U(1), to the case of arbitrary finite groups and compact Lie groups. 


Consider the subspace spanned by $\{T_{m}^{(\mu,\alpha)}:\forall \alpha\}$ for a fixed $m$ and $\mu$. Then lemma \ref{G-cov-Sup-tensor} implies that any G-covariant super-operator maps an operator in this subspace to another operator in this subspace. This suggests the following definition of modes of asymmetry

\begin{definition}\label{Def:Modes of asymmetry}
The $(\mu,m)$ mode component of an operator $X$, denoted $X^{(\mu,m)}$, is defined by
\begin{equation}\label{Def-Mod}
X^{(\mu,m)}\equiv \sum_{\alpha}  {T_{m}^{(\mu,\alpha)}} \text{tr}\left( T_{m}^{(\mu,\alpha)}{ }^{\dag} X \right).
\end{equation}
We call the decomposition $X=\sum_{\mu,m} X^{(\mu,m)}$ the \emph{mode decomposition} of operator $X$.
\end{definition}
Note that in the above definition we have assumed that the basis $\{T_{m}^{(\mu,\alpha)}\}$ is an orthonormal basis, i.e.  $$\text{tr}\left(T_{m}^{(\mu,\alpha)}T_{m'}^{(\mu',\alpha')}{ }^\dag \right)=\delta_{m,m'}\delta_{\mu,\mu'}\delta_{\alpha,\alpha'}.$$

Lemma \ref{G-cov-Sup-tensor} has a simple interpretation in terms of mode decompositions of operators: a G-covariant super-operator $\mathcal{E}$ maps an operator in a particular mode of asymmetry to an operator in the same mode of asymmetry, i.e., if  $Y=\mathcal{E}(X)$ then
\begin{equation}
\forall \mu,m:\ \ Y^{(\mu,m)}=\mathcal{E}(X^{(\mu,m)}).
\end{equation}
So we can think of different pairs $(\mu,m)$ as different independent modes which cannot be mixed under a G-covariant linear super-operator. In particular, if an input $X$ has no component in a particular mode then the corresponding output $Y$ also cannot have any component in that mode.

The above definition is independent of the choice of the tensor operators basis, $\{ T_{m}^{(\mu,\alpha)}\}$. In the following lemma, we present another way to define modes of asymmetry which is explicitly basis independent.

Let $\{g\rightarrow u^{(\mu)}(g)\}$ be the (non-projective) set of all unitary irreps of a finite or compact Lie group $G$ and  $\{u_{mm'}^{(\mu)}(g)\}$ be the matrix elements of these unitary irreps. Recall that these matrix elements satisfy the orthonormality relations,
\begin{equation} \label{mode-orthog-unit}
\int dg\ \bar{u}_{m_{1}m_{2}}^{(\mu)}(g) u_{m_{3}m_{4}}^{(\nu)}(g)=\frac{1}{d_{\mu}}\delta_{\mu,\nu}\delta_{m_{1},m_{3}} \delta_{m_{2},m_{4}},
\end{equation}
where in the case of finite groups the integral is replaced by the summation over all group elements. 
Then one can easily see that the following lemma holds.

\begin{lemma}\label{Mod-decomp-lem}
Let $X=\sum_{\mu,m}X^{(\mu,m)}$ be the mode decomposition of operator $X$. Then
\begin{equation}
X^{(\mu,m)}=d_{\mu}\int dg \  \bar{u}_{mm}^{(\mu)}(g) \ \mathcal{U}_{g}(X)
\end{equation}
where $d_\mu$ is the dimension of the irrep $\mu$, $dg$ is the uniform measure over the group $G$ and the bar represents complex conjugation.
\end{lemma}

\begin{proof}
We start with Eq.~\eqref{eq-def-tens-irrep},
$$\mathcal{U}_{g}(T_{m}^{(\mu,\alpha)})=\sum_{m'}   u_{m'm}^{(\mu)}(g) T_{m'}^{(\mu,\alpha)}. $$
We multiply both sides by $\bar{u}_{nn}^{(\nu)}(g)$ and integrate over $G$. Now we use the orthonormality  relations, Eq.~(\ref{mode-orthog-unit}).
This implies that
\begin{equation}
\forall \alpha:\ \ \ d_{\mu}\int dg\   \bar{u}_{nn}^{(\mu)}(g)\ \mathcal{U}_{g}(T_{m}^{(\mu,\alpha)})= \delta_{m,n} \delta_{\mu,\nu} T_{m}^{(\mu,\alpha)}.
\end{equation}
The lemma follows from this equality together with the definition of mode decompositions given by Eq.~(\ref{Def-Mod}).
\end{proof}

This lemma gives us an alternative method to find the mode decomposition of a given operator.

It is worth emphasizing an important difference between the mode decomposition for the case of non-Abelian groups and the mode decomposition for the case of Abelian groups such as U(1). This difference concerns the result of symmetry transformations on operators in different modes.
Since
\begin{equation}
\mathcal{U}_{g}[T_{m}^{(\mu,\alpha)}]=\sum_{m'}  u_{m'm}^{(\mu)}(g) T_{m'}^{(\mu,\alpha)}
\end{equation}
it follows that modes $(\mu,m)$ and $(\mu',m')$ for which $\mu\neq\mu'$ do not mix together under the action of the group, but modes for which $\mu=\mu'$ and $m\neq m'$  can mix together under this action. This can happen because in general a symmetry transformation $\mathcal{U}_{g}$ is not a G-covariant operation, unless the group $G$ is Abelian. In the Abelian case,  modes are just specified by an irrep label $\mu$.

\subsection{Quantifying the degree of asymmetry in a given mode} \label{sec-modes-monotones}

As we saw in the specific case of  the group U(1), one can quantify, for a given state, the amount of asymmetry in each mode in terms of the trace-norm of the component of the state in that particular mode. By a similar argument, it follows that for each mode $(\mu,m)$ the function defined by 
\begin{equation}
F_{\mu,m}(X)\equiv \|X^{(\mu,m)}\| 
\end{equation}
is an asymmetry monotone.

The constraint on state to ensemble transformations, described in proposition \ref{prop-mode-mon-U(1)-ens} for the U(1) case, generalizes as follows:

\begin{proposition}\label{prop:statetoensembleG}
Suppose there is a G-covariant channel which maps the state $\rho$ to the ensemble containing states $\sigma_{i}$ with probabilities $p_{i}$  where the value of $i$ becomes known at the end of the process. Then 
\begin{equation}
\forall (\mu,m):\ \ \sum_{i} p_{i}\ \| \sigma_{i}^{(\mu,m)} \| \le  \|{\rho}^{(\mu,m)} \|.
\end{equation}
\end{proposition}

Using definition \ref{Def:Modes of asymmetry} we can rewrite this bound as
\begin{align*}
\forall (\mu,m):\  &\sum_{i} p_{i} \sum_{\alpha}  {\text{tr}\left(\sqrt{T_{m}^{(\mu,\alpha)}T_{m}^{(\mu,\alpha)}{ }^\dag} \right)  } \left|\text{tr}\left( T_{m}^{(\mu,\alpha)} \sigma \right)\right|\ \  \\  &\le\ \sum_{\alpha}  {\text{tr}\left(\sqrt{T_{m}^{(\mu,\alpha)}T_{m}^{(\mu,\alpha)}{ }^\dag} \right)  } \left|\text{tr}\left( T_{m}^{(\mu,\alpha)} \rho \right)\right|.
\end{align*}

As a simple corollary of proposition~\ref{prop:statetoensembleG}, if a nondeterministic G-covariant operation maps state $\rho$ to state $\sigma$ with probability $p$, then 
\begin{equation}
\forall (\mu,m):\ \ p\\ \| \sigma^{(\mu,m)} \| \le  \|{\rho}^{(\mu,m)} \|.
\end{equation}

\subsection{Example: spin-$j$ system}
Consider the case of a spin-$j$ representation of SO(3).  Then, as we have seen before, all the modes are multiplicity-free and so
\begin{align*}
\forall (\mu,m):\   F_{\mu,m}(\rho)&\equiv \|\rho^{(\mu,m)} \|\\ &={\text{tr}\left(\sqrt{T_{m}^{(\mu)}T_{m}^{(\mu)}{ }^\dag} \right) } \left|\text{tr}\left( T_{m}^{(\mu)} \rho \right)\right|.
\end{align*}
Now, if a state $\rho$ of the spin-$j$ system evolves under a rotationally invariant dynamics to a state $\sigma_{i}$ of the spin-$j$ system with probability $p_{i}$, then for all modes $ (\mu,m)$ it holds that
\begin{equation}\label{Ex-spinj-modes}
\sum_{i} p_{i}\left|\text{tr}\left(T_{m}^{(\mu)}\sigma_{i} \right)\right| \le \left|\text{tr}\left(T_{m}^{(\mu)} \rho \right)\right|.
\end{equation}
So, for example, in the case of mode $(\mu=1,m=0)$, where $T_{0}^{(1)}=cL_{z}$ for some constant $c$ we find
$$\sum_{i} p_{i}\left|\text{tr}\left(L_{z}\sigma_{i} \right)\right| \le \left|\text{tr}\left(L_{z} \rho \right)\right|.$$
Note that here the direction $\hat{z}$ is chosen arbitrarily and so for any direction $\hat{n}$ it holds that
\begin{equation}\label{bound-spinj-ang}
\sum_{i} p_{i}\left|\text{tr}\left(L_{\hat{n}}\sigma_{i} \right)\right| \le \left|\text{tr}\left(L_{\hat{n}} \rho \right)\right|.
\end{equation}
This result is very intuitive. If a spin-$j$ undergoes a deterministic or stochastic  rotationally-covariant dynamics, the average of the \emph{absolute value} of the expectation value of angular momentum can not increase.  Note that the sign of this expectation value can change, i.e., a state whose angular momentum is negative in the $\hat{z}$ direction can evolve to a state    whose angular momentum is positive in this direction.

In this example we have assumed that the initial and final spaces are both spin-$j$ systems. 
On the other hand, one can easily show that the absolute value of angular momentum can increase if  the final space is allowed to have a higher spin. In the following we will find a bound which applies to the cases where the initial and final spaces have different spins. Before this, we present another consequence of Eq.~(\ref{Ex-spinj-modes}) for the case where both input and output spaces are spin-$j$. 

Although in a rotationally-covariant dynamics of a spin-$j$ system the absolute value of angular momentum  cannot increase, nevertheless the expectation value of higher  powers  of angular momentum can increase. However, using  Eq.(\ref{Ex-spinj-modes}),  we can find non-increasing functions which involve the expectation value of higher powers of angular momentum. For instance, consider the case of $(\mu=2,m=0)$. Then, as we have seen in section \ref{Ex-tensors-SU(2)}, for a spin-$j$ representation of  SO(3),
$$T^{(\mu=2,m=0)}=c (3L_{z}^{2}-L^{2}),$$
where $c$ is a normalization factor. Then  Eq.~(\ref{Ex-spinj-modes}) implies that 
\begin{equation}
p\left| \text{tr}(\sigma L_{z}^{2})-\frac{j(j+1)}{3} \right|\le\left| \text{tr}(\rho L_{z}^{2})-\frac{j(j+1)}{3} \right|,
\end{equation}
where we have used the fact that for all spin-$j$ systems the expectation value of $L^{2}$ is $j(j+1)$. Note that the $\hat{z}$ direction is chosen arbitrarily.  So, for arbitrary direction $\hat{n}$,  $\left| \text{tr}( \rho L_{n}^{2})-{j(j+1)}/{3} \right|$ is non-increasing under rotationally covariant dynamics,  even though $\text{tr}(\rho L_{\hat{n}}^{2})$ can increase.

Now we find a bound on the change of the absolute value of the expectation value of angular momentum when  the input and output spaces have different  spins.
 
To achieve this goal, we calculate $\|\rho^{(\mu,m)} \|$ for the mode $(\mu=1,m=0)$ in the case of a spin-$j$ system. Using the fact that $T^{(\mu=1,m=0)}=c L_{z}$ for some constant $c$,  we find
\begin{align}
F_{\mu=1,m=0}(\rho)\equiv\|\rho^{(1,0)} \|&=\frac{\text{tr}\left(\sqrt{L_{z}^{2}} \right) }{\text{tr}\left(L_{z}^{2}\right)} \left|\text{tr}\left( L_{z} \rho \right)\right|\nonumber \\&=\frac{3\text{tr}\left(\sqrt{L_{z}^{2}} \right) }{\text{tr}\left(L^{2}\right)} \left|\text{tr}\left( L_{z} \rho \right)\right|,
\end{align}
where we have used the normalization condition, i.e. $|c|^{2}\text{tr}\left(L_{z}^{2}\right)=1$. 
One can easily see that $\text{tr}\left({L^{2}} \right)=j(j+1)(2j+1)$ and
\begin{equation}
\text{tr}\left(\sqrt{L_{z}^{2}} \right)=
\left\{
\begin{tabular}[c]{l}
$j(j+1)\ \ \ \ \text{integer}\ j$ \\
$(j+1/2)^{2}\ \ \text{half integer}\ j$
\end{tabular}
\right.
\end{equation}


 So
 \begin{equation}
\|\rho^{(1,0)} \|=
\left\{
\begin{tabular}[c]{l}
$\tfrac{3}{2} \frac{ \left|\text{tr}\left( L_{z} \rho \right)\right|}{j+1/2}\ \ \ \ \  \ \ \ \text{integer}\ j$\\
$\tfrac{3}{2} \frac{ \left|\text{tr}\left( L_{z} \rho \right)\right| (j+1/2)}{j(j+1)}\ \ \text{half integer}\ j$
\end{tabular}
\right.
\end{equation}
So $\|\rho^{(1,0)} \|$ is less than or equal to 3/2 and at the limit of $j$ going to infinity it  tends to $3/2$.

Now we can find an analogue of the bound of Eq.~(\ref{bound-spinj-ang}) for the case where the input and output systems have spins $j$ and $j'$ respectively. If, for example, both $j$ and $j'$ are integer then
\begin{equation}
p \frac{ \left|\text{tr}\left( L_{\hat{n}} \sigma \right)\right|}{j'+1/2}\le \frac{ \left|\text{tr}\left( L_{\hat{n}} \rho \right)\right|}{j+1/2}.
\end{equation}
In proposition \ref{pro-max-succ} below, we show that the quantity $\frac{ \left|\text{tr}\left( L_{z} \rho \right)\right|}{j+1/2}$ admits of an operational interpretation: it quantifies the ability of the state $\rho$ to act as a quantum reference frame for the task of distinguishing, on a spin-1/2 system, the two eigenstates of $L_{z}$, $|j=1/2,m=-1/2\rangle$ and $|j=1/2,m=1/2\rangle$.

\section{Simulating quantum operations by quantum reference frames} \label{Sec:Simulating}

Consider the situation where we are restricted to those Hamiltonians which all have a particular symmetry. Then it is still possible to \emph{simulate} a dynamics which breaks this symmetry if we have access to a state which breaks the symmetry, i.e. a source of asymmetry.   As we have mentioned earlier, this symmetry-breaking state is called a quantum reference frame. By coupling this quantum reference frame to a system via a symmetric dynamics, we can effectively generate an asymmetric dynamics or measurement on this system. In this section we are interested in finding the set of  asymmetric dynamics and measurements which can be simulated using a given quantum reference frame.

As a simple example, consider the case where we are restricted to the rotationally-invariant Hamiltonians. Then by coupling a quantum system to a large magnet with magnetic field in the $\hat{z}$ direction via a rotationally invariant Hamiltonian, we can effectively simulate a rotation around the $\hat{z}$ axis on that quantum system (note that a rotation is not a rotationally invariant operation and so cannot be performed without having access to a system which breaks the rotational symmetry). In this case, we can model the magnet by a spin-$j$ system in a large coherent state polarized in the $\hat{z}$ direction, i.e., in the maximum weight eigenstate of $L_{\hat{z}}$,  $|j,m=j\rangle_{\hat{z}}$. Then, by coupling the quantum system to this quantum reference frame one can realize a quantum channel on the  system such that this channel  at the limit where $j$ goes to infinity approaches a perfect (unitary) rotation. In fact, one can show that using a  spin-$j$ in the coherent state in $\hat{z}$ direction, at the limit of large $j$ any arbitrary dynamics which is invariant under rotation around $\hat{z}$ can be simulated on the system~\cite{Mar-Man}.\footnote{Furthermore, it is shown in Ref.~\cite{Mar-Man} that if in addition one has access to a similar quantum reference frame in a coherent state polarized in the $\hat{x}$ direction then one can simulate arbitrary dynamics on the system.}

Note that by having access to this quantum reference frame  we still cannot  simulate a rotation around $\hat{x}$ or any other dynamics which is not invariant under rotation around $\hat{z}$.  More generally, for a given quantum reference frame, only those time evolutions and measurements can be simulated which have all the symmetries of the quantum reference frame. In this section we generalize this simple observation by finding a relation between the modes of asymmetry  of the quantum reference frame and the modes of asymmetry of a time evolution or measurement that can be simulated using it.

\subsection{Modes of asymmetry of quantum operations} \label{sec:Quant-Mode}
The notion of modes of asymmetry naturally extends to the super-operators. Let $g\rightarrow U(g)$ be the projective unitary representation of $G$ on the Hilbert space $\mathcal{H}$. Also, let $\mathcal{U}_{g}(\cdot)\equiv U(g)(\cdot) U^{\dag}(g)$. Then $g\rightarrow \mathcal{U}_{g}$ is a (non-projective) unitary  representation of $G$ on $\mathcal{B}(\mathcal{H})$. Similarly, we can define a representation of $G$ on the space of all linear super-operators: Consider the linear space of all super-operators from  $\mathcal{B}(\mathcal{H}_{\text{in}})$ to  $\mathcal{B}(\mathcal{H}_{\text{out}})$. Then a natural representation of $G$ on this space is given by the following map
\begin{equation}
\forall g\in G:\ \ \ \  \mathfrak{U}_{g}[\mathcal{E}]\equiv\mathcal{U}^{\text{out}}_{g}\circ\mathcal{E}\circ  \mathcal{U}^{\text{in}}_{g^{-1}}
\end{equation}
for arbitrary $ \mathcal{E}: \mathcal{B}(\mathcal{H}_{\text{in}})\rightarrow \mathcal{B}(\mathcal{H}_{\text{out}})$, where  $g\rightarrow \mathcal{U}^{\text{in}}_{g}$ and $g\rightarrow \mathcal{U}^{\text{out}}_{g}$ are the representations of the symmetry on $\mathcal{B}(\mathcal{H}_{\text{in}})$ and  $\mathcal{B}(\mathcal{H}_{\text{out}})$ respectively. This representation has a natural physical interpretation: suppose the representation $g\rightarrow U(g)$ describes a change of reference frame, such that a state which is described by $|\psi\rangle$ in the old reference frame is described by $U(g)|\psi\rangle$ in the new reference frame. Then, an observable or a density operator which is described by an operator $A$ relative to the old reference frame will be described by $\mathcal{U}_{g}[A]$ relative to the new reference frame.  Similarly, a super-operator which is described by $\mathcal{E}$ relative to the old reference frame will be described by $\mathfrak{U}_{g}[\mathcal{E}]$ relative to the new reference frame.

Now, following the same logic we used to define modes of asymmetry of operators based on  the representation $g\rightarrow \mathcal{U}_{g}$ of group $G$ on the space of operators, we can define the notion of modes of asymmetry of super-operators  based on  the representation $g\rightarrow \mathfrak{U}_{g}$ of group $G$ on the space of super-operators. One way to do this is by defining the analogues of the irreducible tensor operators for super-operators. Alternatively, we can define modes of asymmetry for super-operators using the analogue of lemma \ref{Mod-decomp-lem}:
\begin{definition}\label{def-mode-super}
The mode $(\mu,m)$ of the super-operator $\mathcal{E}$, denoted by $\mathcal{E}^{(\mu,m)}$ is defined by
\begin{equation}
\mathcal{E}^{(\mu,m)}=d_{\mu}\int dg\  \bar{u}_{mm}^{(\mu)}(g)\  \mathfrak{U}_{g}\left[\mathcal{E}\right]
\end{equation}
where $d_{\mu}$ is the dimension of the irrep $\mu$.
We call the decomposition $\mathcal{E}=\sum_{\mu,m} \mathcal{E}^{(\mu,m)}$ the \emph{mode decomposition} of the super-operator $\mathcal{E}$ and $\mathcal{E}^{(\mu,m)}$ the $(\mu,m)$ modal component of $\mathcal{E}$.
\end{definition}

Note that this definition implies that a \emph{G-covariant} super-operator only has nonzero component in the mode which corresponds to the trivial representation of the group, denoted by $\mu=0$.

Let $g\rightarrow U(g)$ be the representation of the symmetry group on the Hilbert space $\mathcal{H}$. As we have seen before, we can find the set of all modes of asymmetry that an operator $X\in\mathcal{B}(\mathcal{H})$ can possibly have, by decomposing the representation $g\rightarrow U(g)\otimes \bar{U}(g)$. Similarly, we can find all modes of asymmetry that a super-operator $ \mathcal{E}: \mathcal{B}(\mathcal{H}_{\text{in}})\rightarrow \mathcal{B}(\mathcal{H}_{\text{out}})$ can possibly have, by decomposing the representation
$$g\rightarrow  U_{\text{out}}(g)\otimes \bar{U}_{\text{out}}(g)\otimes U_{\text{in}}(g)\otimes \bar{U}_{\text{in}}(g)$$
to irreps of $G$.  Here, $g\rightarrow U_{\text{in}}(g)$ and $g\rightarrow U_{\text{out}}(g)$ are the representations of the symmetry group on $\mathcal{H}_{\text{in}}$ and $\mathcal{H}_{\text{out}}$ respectively.
\begin{example}\label{Examp-Modes-Meas}
Consider the group of rotations in $\mathbb{R}^{3}$, i.e. G=SO(3), and assume the input Hilbert space carries a $j_{1}$ irrep and the output Hilbert space carries a $j_{2}$ irrep. Then any super-operator from $\mathcal{B}(\mathcal{H}_{\text{in}})$ to $\mathcal{B}(\mathcal{H}_{\text{out}})$ can have modes $(\mu,m)$ with $\mu\le 2(j_{1}+j_{2})$. In particular, if the input and output spaces of a super-operator are both spin-1/2 systems (i.e. $j_{1}=j_{2}=1/2$), then the super-operator can only have modes $(\mu=0),(\mu=1,m=-1,0,1)$ and $(\mu=2,m=-2,-1,0,1,2)$. On the other hand, if the input space $\mathcal{H}_{\text{in}}$ is a spin-1/2 irrep of SO(3) and the output space is invariant under rotation (i.e. $j_{1}=1/2$ and $j_{2}=0$) then  the super-operator can only have modes $(\mu=0),(\mu=1,m=-1,0,1)$. These latter kind of super-operators can describe, for example,  measurements on a spin-1/2 system where the post-measurement state is always rotationally invariant.
\end{example}
In this example, we  found all modes of asymmetry that a measurement performed on a spin-1/2 system can possibly have. In the following we study the notion of modes of asymmetry of measurements more closely.

\subsubsection{Modes of asymmetry of measurements}
In the  study of the modes of asymmetry of measurements,  we focus on the aspect of a measurement that is relevant for making inferences about the input, that is, its informative aspect, and neglect the aspect that is relevant for making predictions about future measurements on the system, that is, its state-updating aspect. 

Let $\{M_{\lambda}\}$ be the POVM describing an arbitrary measurement. Define the channel
\begin{equation}\label{Meas-chann}
\mathcal{M}(X)\equiv\sum_{\lambda} \text{tr}(X M_{\lambda})|\lambda\rangle\langle\lambda |
\end{equation}
where the set $\{|\lambda\rangle\}$ consists of orthogonal and G-invariant states.
Then, any measurement whose statistics is described by the POVM $\{M_{\lambda}\}$  can be realized by  first applying the channel $\mathcal{M}(\cdot)$ to the state and then measuring the output system in the $\{|\lambda\rangle\}$ basis. But, this latter measurement is  G-covariant and so  one does not need a quantum reference frame to realize it. So to find the asymmetry resources required to implement a measurement of the POVM $\{M_{\lambda}\}$, it suffices to find the asymmetry resources required to implement $\mathcal{M}$.

One can easily show that
\begin{lemma}\label{lem-modes-meas}
$\mathcal{M}^{(\mu,m)}$, the $(\mu,m)$ modal component of $\mathcal{M}$, is equal to
\begin{equation}
\mathcal{M}^{(\mu,m)}(X)=\sum_{\lambda} \text{tr}(X M^{(\mu,m)}_{\lambda}) |\lambda\rangle\langle\lambda |
\end{equation}
where $M^{(\mu,m)}_{\lambda}$ is the $(\mu,m)$ modal component of the operator $M_{\lambda}$.
\end{lemma}
\begin{proof}
First note that by definition \ref{def-mode-super}, 
\begin{equation}
\mathcal{M}^{(\mu,m)}(X)\equiv d_{\mu}\ \int dg\  \bar{u}_{mm}^{(\mu)}(g)\ \mathcal{U}^{\text{out}}_{g}\circ\mathcal{M}\circ  \mathcal{U}^{\text{in}}_{g^{-1}}(X).
\end{equation}
Here, the representation of symmetry on the output system is trivial and so
\begin{align*}
\mathcal{M}^{(\mu,m)}(X)&=d_{\mu}\int dg\  \bar{u}_{mm}^{(\mu)}(g)\ \mathcal{U}^{\text{out}}_{g}\circ\mathcal{M}\circ  \mathcal{U}^{\text{in}}_{g^{-1}}(X)\\ &=d_{\mu}\sum_{\lambda} |\lambda\rangle\langle\lambda | \int dg\  \bar{u}_{mm}^{(\mu)}(g) \text{tr}( \mathcal{U}^{\text{in}}_{g^{-1}}(X) M_{\lambda})\\&=d_{\mu}\sum_{\lambda} |\lambda\rangle\langle\lambda | \int dg\  \bar{u}_{mm}^{(\mu)}(g) \text{tr}\left( X \mathcal{U}^{\text{in}}_{g^{}} (M_{\lambda})\right)\\ &=\sum_{\lambda} \text{tr}(X M^{(\mu,m)}_{\lambda}) |\lambda\rangle\langle\lambda |.
\end{align*}
\end{proof}

Further on, we will use this observation to infer the asymmetry resources which are required to implement a given measurement.

\subsection{From modes of quantum reference frames to modes of quantum channels} \label{Modes of Q.RF}


As described above, under the assumption that all symmetric dynamics are free we can use a quantum reference frame, which breaks the symmetry, as a \emph{resource} of asymmetry which enables us to \emph{simulate} dynamics which break the symmetry.

\begin{definition}
Let $\mathcal{H}_{\text{sys}}$ and $\mathcal{H}_{\text{RF}}$    be Hilbert spaces  with projective unitary representations $g\rightarrow U_{\text{sys}}(g)$ and $g\rightarrow U_{\text{RF}}(g)$ of group $G$.
We say that a channel $\mathcal{E}: \mathcal{B}(\mathcal{H}_{\text{sys}})\rightarrow\mathcal{B}(\mathcal{H}_{\text{sys}})$ can be simulated using  the resource state $\rho_{\text{RF}}$ if  there exists a channel
$\tilde{\mathcal{E}}: \mathcal{B}(\mathcal{H}_{\text{sys}}\otimes \mathcal{H}_{\text{RF}})\rightarrow\mathcal{B}(\mathcal{H}_{\text{sys}}\otimes \mathcal{H}_{\text{RF}})$ which is G-covariant, i.e.,
$$\forall g\in G:\ \ \tilde{\mathcal{E}}\circ (\mathcal{U}_{g}^{\text{sys}}\otimes \mathcal{U}_{g}^{\text{RF}})=(\mathcal{U}_{g}^{\text{sys}}\otimes \mathcal{U}_{g}^{\text{RF}})\circ\tilde{\mathcal{E}}$$
such that
\begin{equation}\label{def-simul}
\mathcal{E}(X)=\text{tr}_{RF}\left( \tilde{\mathcal{E}}(X\otimes \rho_{RF}) \right).
\end{equation}
 \end{definition}

Now one can easily prove the following result.
\begin{lemma}\label{lem-mode-mode}
Suppose the channel $\mathcal{E}$ can be simulated using a quantum reference frame in the state $\rho$ and a G-covariant channel $\tilde{\mathcal{E}}$ such that $\mathcal{E}(X)=\text{tr}_{RF}\left( \tilde{\mathcal{E}}(X\otimes \rho) \right)$.  Then
\begin{equation}
\mathcal{E}^{(\mu,m)}(X)=\text{tr}_{RF}\left( \tilde{\mathcal{E}}(X\otimes \rho^{(\mu,m)}) \right).
\end{equation}
\end{lemma}
\begin{proof}
First, note that
\begin{align*}
\mathfrak{U}_{g}[\mathcal{E}](X)&=\mathcal{U}_{g}\circ\mathcal{E}\circ  \mathcal{U}_{g^{-1}}(X)\\&=\mathcal{U}_{g}\circ \text{tr}_{RF}\left( \tilde{\mathcal{E}}(\mathcal{U}_{g^{-1}}(X)\otimes \rho) \right).
\end{align*}
Then, because $\tilde{\mathcal{E}}$ is G-covariant, we have
\begin{align*}
\mathfrak{U}_{g}[\mathcal{E}](X) &=\text{tr}_{RF}\left( \tilde{\mathcal{E}}(X\otimes \mathcal{U}_{g}(\rho)) \right).
\end{align*}
By multiplying both sides by $\bar{u}_{mm}^{(\mu)}(g)$ and taking the integral over $G$, we prove the lemma.
\end{proof}

In the previous section, we defined $\text{Modes}(\rho)$ to be the set of all modes in which state $\rho$ has nonzero components. Similarly, we define $\text{Modes}(\mathcal{E})$ to be the set of all modes in which a channel $\mathcal{E}$ has nonzero components.
Then the above lemma implies that
\begin{proposition}\label{prop-sim-rf-chan}
If a quantum reference frame $\rho$ can simulate a quantum channel $\mathcal{E}$  then
\begin{equation}
\text{Modes}(\mathcal{E}) \subseteq \text{Modes}(\rho).
\end{equation}
  \end{proposition}
So if a quantum reference frame does not have a particular mode of asymmetry, it cannot simulate a time evolution or measurement which has that mode of asymmetry. Also, the lemma implies that to specify whether a given quantum reference frame $\rho$ can simulate a quantum channel $\mathcal{E}$ or not we only need to know the components of $\rho$ in modes contained in $\text{Modes}(\mathcal{E})$. So, as we will see in an example, although the Hilbert space of the quantum reference frame might be arbitrary large, the number of parameters required to specify its performance for some specific simulation can be very small.

Furthermore, for any given finite dimensional Hilbert space $\mathcal{H}_{\text{sys}}$, there are a finite set of modes  in which a channel acting on $\mathcal{B}(\mathcal{H}_{\text{sys}})$ can have nonzero components. So for a given quantum reference frame $\rho_{\text{RF}}$ on an arbitrarily large Hilbert space and having amplitude over arbitrarily many representations of the group, the properties of the quantum reference frame which are relevant for simulating arbitrary channels acting on $\mathcal{B}(\mathcal{H}_{\text{sys}})$ can be specified merely by specifying the components of $\rho_{\text{RF}}$ in that finite set of modes.

\subsubsection*{Example: Reference frames of unbounded size may still lack modes}
In the case of $G$=SO(3),  consider the family of quantum reference frames defined by
$$|\psi_{N}\rangle\equiv \frac{1}{\sqrt{N}}\sum_{k=1}^{N} |{j=N^2+2 k ,m=N^2+k }\rangle.$$
One can easily show that, at the limit of large $N$ these states are very sensitive to rotations around $\hat{z}$ and also rotations around any axis in the $\hat{x}-\hat{y}$ plane. In other words, for any small such rotation, $|\psi_{N}\rangle$ is almost orthogonal to the rotated version of $|\psi_{N}\rangle$. So, one may think that at the limit of large $N$ this quantum reference frame completely breaks the symmetry and so it can be used to simulate any arbitrary measurement on a spin-1/2 system.  This is not the case, however. Indeed, it turns out that even though at the limit of large $N$ these states are very sensitive to rotations around $\hat{z}$ (that is, they break this symmetry), nonetheless they cannot simulate any measurement  on a spin-1/2 system which is not invariant under rotations around $\hat{z}$. To see this, first note that if the POVM elements of a measurement on a spin-1/2 system are not invariant under rotations around $\hat{z}$ then they have nonzero components in the modes $(\mu=1,m=\pm 1)$. Then, lemma \ref{lem-modes-meas} implies that the channel describing that measurement will have modes $(\mu=1,m=\pm 1)$. Now proposition \ref{prop-sim-rf-chan} implies that to be able to simulate  any such measurement a quantum reference frame needs to have a non-zero component in the modes $(\mu=1,m=\pm 1)$. 
However, using the Wigner-Eckart theorem, one can easily show that  none of the states in the above family have a nonzero-component in the modes $(\mu=1,m=\pm 1)$\footnote{Consider the terms in the decomposition 
\begin{align*}
\left|\psi_{N}\rangle\langle\psi_{N}\right|&= \frac{1}{{N}}\sum_{k,k'=1}^{N}\\ &\left|{j=N^2+2 k ,m=N^2+k }\rangle\langle {j'=N^2+2 k' ,m'=N^2+k' }\right|
\end{align*}
Any term in this decomposition with $k=k'$ is invariant under rotations around $\hat{z}$ and so it only has components in modes $(\mu,m=0)$. On the other hand, terms with $k\neq k'$ have no components in any of the modes $(\mu=1, m=-1,0,+1)$ because $j$ and $j'$ always differ by at least 2.  
}. 

The conclusion is that there are measurements that break rotational symmetry that cannot be simulated by this family of quantum reference frames.



\section{Example: Spin-$j$ system as a quantum reference frame} \label{sec:Spin-j}

The problem of using a spin-$j$ system as a quantum reference frame to simulate dynamics or measurements which are not invariant under rotations has been studied in several papers (See e.g. \cite{BRST04, BRST06b, BRST06, BRS07, Poulin-Yard, Mar-Man,Sher-Bart, BIM, Ahm-Rud,MarvianSpekkensWAY}). In this section, we 
show that the mode decomposition of states provides an extremely powerful insight into this problem. In particular,  we show that  using this approach some of the previously known results which have been  achieved in an ad hoc manner can be reproduced and generalized in a systematic way.

\subsection{Simulating measurements and dynamics on a spin-1/2 system}

We start with the problem of simulating measurements on a spin-1/2 system. Here, the assumption is that we are restricted to use rotationally-invariant unitaries, ancillary systems in rotationally-invariant states  and measurements whose POVM elements are all rotationally-invariant. Under this restriction,   we  are given a spin-$j$ system in an arbitrary state $\rho$ as a quantum reference frame  and our goal is to simulate a non-invariant measurement on a spin-1/2 system. Here, we focus on the informative aspect of the measurement, that is, we are not concerned with how the state of system is updated after the simulated measurement. 

For an arbitrary measurement on a spin-1/2 system, consider the  channel which describes  the informative aspects of   this measurement, as defined in Eq. (\ref{Meas-chann}). Then, consider the set of all modes $\{(\mu,m)\}$ in which this channel can have nonzero components. We can conclude from lemma \ref{lem-modes-meas}  that this set is equal to  $\{(\mu=0),(\mu=1,m=-1,0,+1)\}$ (This is also shown in example \ref{Examp-Modes-Meas}.).

 Then, it follows from proposition \ref{lem-mode-mode} that the only properties of $\rho$ that are relevant for its performance in simulating a measurement on a spin-1/2 system are its components in the modes $\{(\mu=0),(\mu=1,m= -1,0,+1)\}$, i.e.,   $\rho^{(\mu=0)}$ and $\{ \rho^{(\mu=1,m)}:m=-1,0,1\}$.  Furthermore, since the irreducible tensor operator basis on a spin-$j$ system is multiplicity-free, each of these components
is determined by only one parameter, namely, the Hilbert Schmidt inner product between $\rho$ and the corresponding component of the irreducible tensor operator basis,
 $$T^{(\mu=0)}=c_{0} \mathbb{I},\ \  \ T_{m=0}^{(\mu=1)}=c_{1} L_{z},\ \ \  \text{and}\ \ \ T_{m=\pm 1}^{(\mu=1)}=\pm \frac{c_{1}}{\sqrt{2}} L_{\pm},$$ 
where $c_{0}$ and $c_{1}$ are normalization factors. It follows that the components of $\rho$ in the modes  $\{(\mu=0),(\mu=1,m= -1,0,+1)\}$ are uniquely specified by the vector of expectation values of angular momentum for $\rho$, i.e.,
$\left(\langle L_{x}\rangle, \langle L_{y}\rangle, \langle L_{z}\rangle\right).$

So we conclude that 
\begin{proposition}\label{prop-spin-j}
Consider a spin-$j$ system in state $\rho$ as a quantum reference frame.  Its performance in simulating (informative aspects of) measurements on a spin-1/2 system is uniquely specified  by 
the vector of expectation values of angular momentum of $\rho$.
\end{proposition}
This result has been previously obtained in \cite{Poulin-Yard} using a totally 
different and rather ad hoc argument.  But, using our approach we can easily generalize it to the case of measurements on systems supporting arbitrary representations of SO(3) (as opposed to just the spin-1/2 representation). 
Before presenting this generalization we investigate some implications of proposition \ref{prop-spin-j}.

An interesting consequence 
is the following: Suppose the system is spin-$j$ and the vector of expectation values of angular momentum of $\rho$ is polarized in the $\hat{n}$ direction. In general, for $j \ne 1/2$,  such a state need not be invariant under rotations around $\hat{n}$. Now consider the symmetrized version of $\rho$, i.e., the state
$$\rho_{\text{sym}}\equiv\frac{1}{2\pi}\int d\theta\ e^{-i\theta \vec{L}.\hat{n}} \rho\ e^{i\theta \vec{L}.\hat{n}},$$
which \emph{is} invariant under rotations around $\hat{n}$. One can easily see that this state has the same vector of expectation values of angular momentum as the original state $\rho$. Therefore, any measurement on the spin-1/2 system which can be simulated using $\rho$ can also be simulated using $\rho_{\text{sym}}$ and vice versa. But, since   $\rho_{\text{sym}}$ is invariant under rotations around $\hat{n}$, it can only simulate those measurements whose POVM elements are invariant under rotations around $\hat{n}$. This argument implies that

\begin{corollary}
Consider a spin-$j$ system in state $\rho$ as a quantum reference frame for direction.  If $\rho$ has a vector of expectation values of angular momentum that is polarized in the $\hat{n}$ direction, then it can only simulate those measurements on a spin-1/2 system for which the POVM elements are invariant under rotations around $\hat{n}$.
\end{corollary}
So, even at the limit of large $j$, a single spin-$j$ system  cannot act as a perfect reference frame for direction  even if it does not have any symmetries. 


As an example of simulating measurements on a spin-1/2 system consider the following problem: suppose one uses a spin-$j$ system in the state $\rho$ as a quantum reference frame to  measure the angular momentum of a spin-1/2 system in the $\hat{z}$ direction, that is, to measure the  observable $\sigma_{z}$. It can be shown that this measurement cannot be simulated perfectly with a quantum reference frame of bounded size.  This result is known as the Wigner-Araki-Yanase theorem~\cite{Wigner,ArakiYanase,LoveridgeBusch}, and is quite an intuitive result in the context of the resource theory of asymmetry~\cite{MarvianSpekkensWAY,Ahm-Rud}. Now the question is: using  the  spin-$j$ system in the state $\rho$ as a quantum reference frame, how well can one simulate this measurement? In other words,  using this quantum reference frame, what is the highest  precision attainable in a simulation of a measurement of $\sigma_{z}$?  We evaluate the precision of the realized measurement using the following figure of merit: the highest probability of successfully distinguishing an unknown eigenstate  of $\sigma_{z}$ when we are given each of the two eigenstates  with equal probability.

The answer, which we prove in the appendix, is as follows.
 \begin{proposition}\label{pro-max-succ}
Suppose we are restricted to the rotationally invariant measurements but we have access to the state $\rho$ of a spin-$j$ system as a quantum reference frame. Then the  maximum probability of successfully distinguishing the two eigenstates of $\sigma_{z}$ for a spin-1/2 system, $|j=1/2,m=1/2\rangle$ and $|j=1/2,m=-1/2\rangle$, when the prior probabilities of the two are equal,   is given by 
\begin{equation}
p_{\text{succ}}(\rho)=\frac{1}{2}\left[1+\frac{\left|\text{tr}(\rho L_{z})\right|}{j+1/2}\right].
\end{equation}
\end{proposition}
So, as we expected from proposition \ref{prop-spin-j}, this probability only depends on the expectation value of the vector of angular momentum (in this case, just the $\hat{z}$ component).  Note that, as one may expect intuitively, in the limit where $j$ becomes arbitrarily large, the coherent state $|j,j\rangle$ can be used as  a quantum reference frame to simulate the measurement of $\sigma_{z}$ perfectly.


Finally, it is worth mentioning that if the Hilbert space of the quantum reference frame under consideration carries different irreps of SO(3) or if it has more than one copy of an irrep, then the vector of angular momentum of the reference frame state alone is not sufficient to specify the measurements it can simulate on a spin-1/2 system. For example, suppose the quantum reference frame is formed from a spin-$j$ system and a qubit whose states are invariant under rotation. This means that the total Hilbert space of the quantum reference frame has two copies of irrep $j$ of SO(3). Suppose the quantum reference frame is in the state
$$\frac{1}{\sqrt{2}}\left(|j,m=j\rangle|\lambda=1\rangle+|j,m=-j\rangle|\lambda=2\rangle\right)$$
 where $\lambda$ labels different orthogonal states of the qubit.   Then, one can easily show that at the limit of large $j$ this reference frame can simulate any arbitrary measurement which is invariant under rotation around $\hat{z}$. But, the expectation value of angular momentum for this state is zero in all directions.  So, for a general representation of SO(3) these expectation values cannot characterize the ability of the state to simulate measurements on a spin-1/2 system.

Proposition \ref{prop-spin-j} can be easily generalized to the problem of simulating arbitrary \emph{dynamics} on a spin-1/2 system.\footnote{This includes, as a special case, the problem of simulating measurements when we are concerned with simulating not just their informative aspect, but the particular update rule as well.} 

Recalling example \ref{Examp-Modes-Meas}, it is clear that the modes of asymmetry appearing in the modal decomposition of a channel on a spin-1/2 system are $\{(\mu=0),(\mu=1,m=-1,0,1), (\mu=2,m=-2,-1,0,1,2)\}$.
 So, proposition \ref{lem-mode-mode} implies that to specify the ability of a particular state $\rho$ of the spin-$j$ system to simulate an arbitrary dynamics on a spin-1/2 system, we merely need to specify these modes of asymmetry of $\rho$. Again from the result of  section \ref{Ex-tensors-SU(2)} we can see that the components of $\rho$ in these modes are uniquely specified by the following eight real parameters: 
\begin{align} \label{eight-par}
&\mu=1\ \text{modes}:\ \langle L_{x}\rangle, \langle L_{y}\rangle, \langle L_{z}\rangle \nonumber\\
&\mu=2\ \text{modes}:\  \langle L^{2}_{x}\rangle, \langle L^{2}_{y}\rangle, \ \langle L_{x}L_{y}+L_{y}L_{x}\rangle, \nonumber \\ &\langle L_{x}L_{z}+L_{z}L_{x}\rangle,\ \ \langle L_{y}L_{z}+L_{z}L_{y}\rangle
\end{align}
So to summarize we have shown that
\begin{proposition}\label{prop-spin-j-chan}
Consider a spin-$j$ system in state $\rho$ as a quantum reference frame.  Its performance in simulating channels on a spin-1/2 system is uniquely specified  by the eight  real parameters specified in Eq.(\ref{eight-par}). 
\end{proposition}

\subsection{Generalization to arbitrary systems}
In the previous section we considered the problem of simulating measurements and  channels on spin-1/2 systems. In this section we generalize these results to arbitrary systems. We start by a generalization of propositions \ref{prop-spin-j} and \ref{prop-spin-j-chan}.

 \begin{theorem}\label{prop-gen-spin-l}
Consider a  spin-$j$ system  in state $\rho$ as a quantum reference frame.  Its performance for simulating a measurement (respectively channel) on a system with Hilbert space $\mathcal{H}$ is uniquely specified  by $(2l+1)^{2}-1$ (respectively $(4l+1)^{2}-1$) real parameters, where $l$ is the largest angular momentum quantum number which  shows up in the decomposition of $\mathcal{H}$ into irreps of SO(3). These parameters of $\rho$ correspond to the expectation values of  all the non-trivial irreducible tensor operators with rank less than or equal to $2l$ (respectively $4l$). 
\end{theorem}
The proof is presented in appendix \ref{app:proofs}. Note that an arbitrary state $\rho$ of a spin-$j$ system is specified by $(2j+1)^{2}-1$ real parameters. But the above result implies that the number of parameters of $\rho$ that are relevant for its performance in the simulation task depends only on $l$, the size of the system to which the measurement or channel is applied, and not on $j$, the size of the quantum reference frame.

An important special case is where the state $\rho$ is invariant under rotations around some axis $\hat{n}$. This special case has been previousy considered for example in \cite{Sher-Bart}. It follows from theorem \ref {prop-gen-spin-l} that
\begin{corollary} \label{cor-prop-gen-spin-l}
Consider a  spin-$j$ system  in state $\rho$ as a quantum reference frame, and suppose that $\rho$  
is invariant under rotations around the direction $\hat{n}$. Its performance in simulating a measurement (respectively channel) can be uniquely specified by $2l$ (respectively $4l$) real parameters corresponding to the moments $\{\text{tr}\left(\rho L^{k}_{\hat{n}}\right): 1 \le k\le 2l\}$ (respectively $\{\text{tr}\left(\rho L^{k}_{\hat{n}}\right): 1 \le k\le 4l\}$) where  $L_{\hat{n}}$ is the angular momentum operator in the $\hat{n}$ direction.
\end{corollary}

Note that to specify  an arbitrary state $\rho$ of a spin-$j$ system which has the relevant symmetry property (invariance under rotations around direction $\hat{n}$), one needs $2j$ real parameters.  An instance of these parameters are $\{\text{tr}\left(\rho L^{k}_{\hat{n}}\right): 1 \le k\le 2j\}$. This particular characterization of states is used in \cite{Sher-Bart}  to specify how the quality of a quantum reference frame \emph{degrades} after using it to simulate a  channel or measurement. 

 In particular, they use this characterization to study the problem of simulating channels on a spin-1/2 system and of simulating measurements  on a spin-1 system.  But from the result of corollary  \ref{cor-prop-gen-spin-l}  we know that  to specify the performance of the quantum reference frame in both of these cases, we only need to specify two real parameters, namely, $\text{tr}\left(\rho L_{\hat{n}}\right) $ and $\text{tr}\left(\rho L^{2}_{\hat{n}}\right)$. As we will see next, this can lead to a significant simplification of the problem of characterizing how a quantum reference frame degrades when used to implement measurements and channels on another system.

\subsection{Degradation of quantum reference frames}
 Using a quantum reference frame to simulate a symmetry-breaking measurement or channel will inevitably \emph{degrade} it. This \emph{degradation} of quantum reference frames can be understood as a manifestation of the fact that obtaining information about a quantum system will necessarily disturb it.
 
  For example, in the case of rotational symmetry, consider a quantum reference frame which specifies an unknown direction in space.  We can use this quantum reference frame to simulate a rotation around this unknown direction on an object system. But by comparing the initial and final state of the object system we can  obtain some information about the unknown direction. So, using a quantum reference frame for simulating a rotation can be understood as performing a measurement on the quantum reference frame and since we thereby obtain some information about the quantum reference frame, its state is necessarily  disturbed in the process.

Different aspects of the degradation of quantum reference frames have been studied in  several papers (See e.g. \cite{Aha-Kauf,Aha-Popes, BRST06, BRST06b, Poulin-Yard, Sher-Bart,Ahm-Rud,MarvianSpekkensWAY} and the references in \cite{BRS07}). A central question studied in these papers is how 
the performance of the quantum reference frame for simulating the measurement or channel drops as a function of the number of implementations of the latter.

A natural special case of the degradation problem, 
considered in \cite{BRST06, BRST06b,Sher-Bart},  is where the average of the state of the system to which the measurement or channel is applied is symmetric. In other words, each time we use the quantum reference frame to simulate an operation on the object system, the initial state of the object system  is chosen at random from an ensemble the average state of which does not break the symmetry.  So, for example, in the case of rotational symmetry, which we study in this section,  the average state of the object system is assumed to be rotationally-invariant.

Then it follows that under this assumption the degradation of the quantum reference frame will be described by a rotationally-covariant channel.  In other words, the state of the quantum reference frame after $k$ uses, denoted $\rho_k$,
will be 
\begin{equation}
\rho_k=\mathcal{E}_{\text{Deg}}(\rho_{k-1})
\end{equation}
where $\mathcal{E}_{\text{Deg}}$ is a G-covariant channel. 
 This implies that under this assumption about the distribution of the states of the object system, different modes of asymmetry of the quantum reference frame degrade independently, i.e., 
\begin{equation}
\forall(\mu,m):\ \  \rho^{(\mu,m)}_k=\mathcal{E}_{\text{Deg}}(\rho^{(\mu,m)}_{k-1})
\end{equation}
This simple observation can greatly simplify the analysis.

Consider the case of spin-$j$ quantum reference frames for direction. First, note that this observation together with theorem   \ref{prop-gen-spin-l}  implies that  the quality of simulation of a channel or measurement on an object  system  after using the quantum reference frame for arbitrary number of times only depends on a fixed number of parameters of the initial state of the quantum reference frame and this number is  independent of the size of the quantum reference frame.
  
Furthermore, from example \ref{Ex-Rot-cov} we know that since the channel $\mathcal{E}_{\text{Deg}}$ which describes the degradation of the quantum reference frame is rotationally covariant it holds that
 \begin{equation}
\forall(\mu,m):\ \  \rho^{(\mu,m)}_k=c^{{(\mu)}}  \rho^{(\mu,m)}_{k-1}
\end{equation}
where $\{c^{(\mu)}\}$ is a set of real coefficients which describe the channel $\mathcal{E}_{\text{Deg}}$ and $\forall \mu: \left|c^{(\mu)}\right|\le1$.
But since, for a spin-$j$ representation of SO(3), the elements of the irreducible tensor operator basis $\{T_{m}^{(\mu)} \}$ are multiplicity-free it holds that $\rho^{(\mu,m)}_k=\text{tr}\left(\rho_{k} T_{m}^{(\mu)}{ }^{\dag} \right)T_{m}^{(\mu)}$ and therefore
 \begin{equation}
\forall(\mu,m):\ \  \text{tr}\left(\rho_{k} T_{m}^{(\mu)}{ }^{\dag} \right)=c^{{(\mu)}} \text{tr}\left(\rho_{k-1} T_{m}^{(\mu)}{ }^{\dag} \right)
\end{equation}
So if $\rho$ and $\rho_{k}$ are respectively the  initial state of the quantum reference frame and its state after $k$ uses,
then it holds that
 \begin{equation}\label{degrade-exp}
\forall(\mu,m):\ \  \text{tr}\left(\rho_{k} T_{m}^{(\mu)}{ }^{\dag} \right)=\left(c^{(\mu)}\right)^{k} \text{tr}\left(\rho T_{m}^{(\mu)}{ }^{\dag} \right)
\end{equation}
Since $|c^{(\mu)}|\le 1$ we can conclude that each of the modal components of the state of the quantum reference frame either remains constant or decays exponentially.

\begin{example}
Here, we consider the scenarios studied in \cite{Sher-Bart} where a spin-$j$ system is used as a quantum reference frame to simulate channels on a spin-1/2 system and measurements on a spin-one system.  Furthermore, it is also assumed that the average of the state of the object system is rotationally-invariant. This implies that the channel which describes the degradation of the quantum reference frame is also rotationally-covariant. It is also assumed in \cite{Sher-Bart} that the state of the quantum reference frame is initially invariant under rotation around an arbitrary direction which we denote by $\hat{z}$. Note that since the degradation of the quantum reference frame is described by a rotationally-covariant channel, the state of the quantum reference frame will remain invariant under rotations around $\hat{z}$. 

Now from theorem \ref{prop-gen-spin-l} and corollary \ref{cor-prop-gen-spin-l} we know that the performance of the state  $\rho$ of this quantum reference frame for these simulations is uniquely specified by two real parameters: the components of $\rho$ in modes $(\mu=1,m=0)$ and  $(\mu=2,m=0)$. But these are specified by   
\begin{align*}
\text{tr}\left(\rho T_{m=0}^{(\mu=1)}{ }^{\dag} \right)&= A _{1}\text{tr}\left(\rho L_{\hat{z}}\right),\  \\ \text{and}\ \ \ \    \text{tr}\left(\rho T_{m=0}^{(\mu=2)}{ }^{\dag} \right)&=A_{2}\text{tr}\left(\rho (3L^{2}_{z}-L^{2})\right)
\end{align*} 
where $A_{1}$ and $A_{2}$ are independent of $\rho$.\footnote{ $A^{-1}_{1}=\sqrt{\text{tr}(L^{2}_{z})}$ and $A^{-1}_{2}=\sqrt{\text{tr}([3L^{2}_{z}-L^{2}]^{2})}$.} Note that since the state $\rho$ by assumption is confined to the  irrep $j$ of SO(3), it follows that $\text{tr}(\rho L^{2})=j(j+1)$ and so 
$$ \text{tr}\left(\rho T_{m=0}^{(\mu=2)}{ }^{\dag} \right)=A_{2}\left[3\text{tr}\left(\rho L^{2}_{z}\right)-j(j+1)\right] $$
In other words, the quality of simulation is uniquely specified by the expectation values of the first and the second moments of $L_{z}$ for $\rho$.
  
Now using  Eq.~(\ref{degrade-exp}) we can conclude that if the initial state of the quantum reference frame is $\rho$ and if we have used the quantum reference frame $k$ times then the quality of the $(k+1)$th simulation is uniquely specified by
\begin{equation}\label{decay-ang}
\text{tr}(\rho_{k} L_{z})= \left(c^{(1)}\right)^{k}\text{tr}\left(\rho L_{z}\right)
\end{equation}
and 
\begin{equation}\label{decay-ang-sq}
\text{tr}(\rho_{k} L^{2}_{z})= \left[c^{(2)}\right]^{k}\text{tr}\left(\rho L^{2}_{z}\right)+\left[1- \left(c^{(2)}\right)^{k}\right] \frac{j(j+1)}{3} 
\end{equation}
where $\{c^{(\mu)}\}$ is the set of coefficients which describe the degradation channel $\mathcal{E}_{\text{Deg}}$. So, in the example studied in \cite{Sher-Bart}, the only properties of the channel $\mathcal{E}_{\text{Deg}}$ which  
are relevant to specify the drop in the quality of simulation   after 
many uses
are the two real coefficients $c^{(1)}$ and  $c^{(2)}$. Finally, note that  since $|c^{(1)}|\le 1$ then Eq.(\ref{decay-ang}) implies  that the absolute value of $\text{tr}(\rho_{k} L_{z})$  is either constant or decays exponentially with $k$. Similarly,  since $|c^{(2)}|\le 1$ then Eq.(\ref{decay-ang-sq}) implies  that $\text{tr}(\rho_{k} L^{2}_{z})$ is either constant or exponentially saturates to $ {j(j+1)}/{3}$, which is the expectation value of $L^{2}_{z}$ for the completely mixed state.
\end{example}

\section{Acknowledgements}
We acknowledge helpful discussions with Gilad Gour. Also, RWS acknowledges an early discussion with Matthias Christandl on the topic of simulating channels with quantum reference frames. Research at Perimeter Institute is supported in part by the Government of Canada through NSERC and by the Province of
Ontario through MRI. IM acknowledges support from NSERC, a Mike and Ophelia Lazaridis fellowship, and ARO MURI grant W911NF-11-1-0268.

\appendix

\section{Quantum coherence as asymmetry relative to phase shifts}\label{App:coherence}
We here argue that quantum coherence, considered as a resource, is simply the resource of asymmetry relative to phase shifts.  

Consider the case of coherence between the eigenspaces of a number operator $N$. 
For the phase complementary to $N$, the group of phase shifts is represented by the set of unitaries  $\{e^{i\theta N}: \theta\in (0,2\pi]\}$.
Clearly, the set of states that are invariant under phase shifts are those that are block-diagonal relative to the eigenspaces of the number operator, i.e., precisely those that have no coherence relative to these eigensapces.  Hence, a state with coherence is one that has some asymmetry relative to phase shifts.

Another observation which supports this view of coherence is that phase-insensitive operations, that is, operations that commute with all phase shifts, cannot generate coherence. In other words, under a phase-insensitive time evolution, incoherent states such as $|n\rangle\langle n|$ or $\frac{1}{{2}}(|0\rangle\langle 0|+|n\rangle\langle n|)$ cannot evolve to states which have coherence. From the point of view of the theory of asymmetry, this is a special example of the more general fact that for any given symmetry, symmetric time evolutions cannot take symmetric states to asymmetric states. Therefore, the problem of quantifying and classifying coherence can be considered as a special case of the theory of asymmetry where the group under consideration is U(1).   The theory of asymmetry is more general, however, because it also concerns non-Abelian groups, such as SO(3), where there is no preferred set of subspaces coherences between which imply asymmetry.

A slightly different approach to the study of coherence was  recently proposed  in \cite{Plenio}. Instead of focusing on phase-insensitive operations, the focus is on a class of \emph{incoherent operations}, which is defined as the set of all operations that transform every incoherent state to another incoherent state. One can easily show that incoherent operations are a proper subset of phase-insensitive operations: For instance, for any $n\neq m$ and $n'\neq m'$, the transformation $\alpha |n\rangle+\beta |m\rangle\longrightarrow\alpha |n'\rangle+\beta |m'\rangle$  can be realized via an incoherent operation while it is forbidden under phase-insensitive operations unless $n-m=n'-m'$. In other words, incoherent operations allow arbitrary permutations among the number eigenspaces. 
It follows that all states of the form $\alpha |n\rangle+\beta |m\rangle$ for $n\neq m$ are equivalent (i.e. reversibly interconvertable) under incoherent operations, whereas they are not all equivalent relative to phase-insenstive operations. 

To decide which of these two subsets of quantum operations is best suited to define coherence as a resource,
one should consider whether there is some realistic restriction on experimental capabilities  that would imply the realizability of only this subset.  In other words, we ask whether one can provide an operational interpretation of either subset.  Just an interpretations exists for phase-insensitive operations: these arise from the restriction that is imposed on a pair of parties when they lack a shared phase reference~\cite{BRS07}.  For instance, if the phase describes the configuration of an oscilator that is acting as a clock, then if two parties fail to have synchronized clocks, they lack a shared phase reference.
In these situations  the transformation $\alpha |n\rangle+\beta |m\rangle\longrightarrow\alpha |n'\rangle+\beta |m'\rangle$ cannot happen unless $n-m=n'-m'$. On the other hand, it is not clear if there is any operational scenario which motivates the study of incoherent operations. 

This analysis is bolstered by considering a similar distinction in entanglement theory.  
The subset of quantum operations that is taken to define the resource of entanglement is the set of Local Operations and Classical Communication (LOCC).  This is distinct from the set of  \emph{non-entangling} operations, the operations which map unentangled states to unentangled states. 
The latter set in particular includes nonlocal operations such as swapping two separated systems.  There is no obvious restricton on experimental capabilities that would permit swapping without also permiting the use of a quantum channel.  In other words, while the set of LOCC operations has a clear operational meaning, the set of nonentangling operations does not. 
Nonlocal operations such as swap are the counterpart, within the set of nonentangling operations, of the general permutations within the set of incoherent operations: neither seems to admit of a good operational interpretation.

\section{Proofs}\label{app:proofs}

\subsection{Proof of lemma \ref{G-cov-Sup-tensor}}

Since  $\{S_{m}^{(\mu,\alpha)}\}$ is a basis for $\mathcal{B}(\mathcal{H}_{out})$ then for any map $\mathcal{E}$
\begin{equation} \label{proof-lem-tensor}
\mathcal{E}(T_{m}^{(\mu,\alpha)})=\sum_{\mu',m',\beta} c_{(\mu,\mu';m,m';\alpha,\beta)} S_{m'}^{(\mu',\beta)}
\end{equation}
for some coefficients $c(\mu,\mu';m,m';\alpha,\beta)$. Now we apply the super-operator $\mathcal{U}_{g}$ to both sides of the above equation. Applying $\mathcal{U}_{g}$ on the left hand side and using G-covariance of $\mathcal{E}$ we get
\begin{equation}
\mathcal{U}_{g} (\mathcal{E}(T_{m}^{(\mu,\alpha)}))=\mathcal{E}(\mathcal{U}_{g}(T_{m}^{(\mu,\alpha)}))=\sum_{m''} u^{(\mu)}_{m''m}(g) \ \mathcal{E}(T_{m''}^{(\mu,\alpha)})
\end{equation}
On the other hand, applying $\mathcal{U}_{g}$ to the right-hand side of Eq.(\ref{proof-lem-tensor}) we get
\begin{align*}
\mathcal{U}_{g}(\sum_{\mu',m',\beta} &c_{(\mu,\mu';m,m';\alpha,\beta)} S_{m'}^{(\mu',\beta)})=\\ &\sum_{\mu',m',\beta} c_{(\mu,\mu';m,m';\alpha,\beta)} \sum_{m''} u^{(\mu')}_{m''m'}(g) S_{m''}^{(\mu',\beta)}
\end{align*}
Equating the right hand sides of the above two equations and using the orthogonality of the functions  $\{u_{mm''}^{(\mu)}(g)\}$ we find that $c_{(\mu,\mu';m,m';\alpha,\beta)}$ can be written as
\begin{equation}
c_{(\mu,\mu';m,m';\alpha,\beta)}=\delta_{mm'}\delta_{\mu\mu'} c_{\beta\alpha}^{(\mu)}
\end{equation}
So we conclude that
\begin{equation}
\mathcal{E}(T_{m}^{(\mu,\alpha)})=\sum_{\beta} c_{\beta\alpha}^{(\mu)}  S_{m}^{(\mu,\beta)}
\end{equation}
Note that  the orthonormality of the basis  $\{S_{m}^{(\mu,\alpha)}\}$ implies
\begin{equation}
c_{\beta\alpha}^{(\mu)}= \text{tr}\left( S_{m}^{(\mu,\beta)}{ }^{\dag}   \mathcal{E}(T_{m}^{(\mu,\alpha)})\right)
\end{equation}
which holds for all $m$. Finally, we notice that the orthonormality of the basis  $\{T_{m}^{(\mu,\alpha)}\}$ together with linearity of  $\mathcal{E}$ implies
\begin{equation}
\mathcal{E}(X)= \sum_{\mu,m,\alpha}  \text{tr}\left( T_{m}^{(\mu,\alpha)}{ }^{\dag} X \right)  \mathcal{E}(T_{m}^{(\mu,\alpha)})
\end{equation}
These last three equations together prove the lemma.

\subsection{Proof of Theorem \ref{prop-gen-spin-l}} 
We start by the proof in the case of  measurements.  
This proof follows  exactly the same as the proof in the special case of spin-1/2 systems. Suppose $\mathcal{H}$ is the Hilbert space of the system on which we simulate the measurement. Then by assumption the largest irrep of SO(3)
 showing up in  $\mathcal{H}$ is $l$. Let $g\rightarrow U(g)$ denote the projective representation of SO(3) on $\mathcal{H}$.
 
Then, as we have seen in section \ref{sec-rev} the set of possible ranks of the irreducible tensor operators acting on $\mathcal{H}$ is  the same as the set of all irreps of SO(3) which show up in the representation $g\rightarrow U(g)\otimes \bar{U}(g)$. But from section \ref{Ex-tensors-SU(2)} we know that in the case of SO(3) this set is equal to the set of all irreps which show up in the representation $g\rightarrow U(g)\otimes {U}(g)$. Now since the maximum irrep in the representation $g\rightarrow U(g)$ is $l$ then the maximum irrep in the representation $g\rightarrow U(g)\otimes {U}(g)$ is $2l$. Therefore, we conclude that the maximum rank of an irreducible tensor operator acting on $\mathcal{H}$ is $2l$. 

Now from lemma \ref{lem-modes-meas} we know that the channel describing the informative aspect of an arbitrary measurement on this space  has mode in the set $\{(\mu,m): \ \mu\le 2l\}$. This together with lemma \ref{lem-mode-mode} implies that to specify the performance of state $\rho$ of quantum reference frame to simulate measurements on this system  we only need to specify the  components of $\rho$ in all modes with $\mu\le 2l$. But, since  the irreducible tensor operators acting on the space of spin-$j$ system have no multiplicity there is only
$$\sum_{k=0}^{2l} (2k+1)=(2l+1)^{2}$$
independent irreducible tensor operators with rank less than or equal to $2l$. Furthermore, the rank $0$ tensor operator is proportional to the identity operator and so the component of $\rho$ in this mode is fixed by the normalization. This implies that  the performance of the quantum reference frame is determined by specifying at most $(2l+1)^{2}-1$ complex numbers corresponding to the expectation values of the density operator $\rho$ for all non-trivial irreducible tensor operators with rank less than or equal to $2l$ . Furthermore,  in the case of SO(3),  as we have seen in the discussion after Eq.(\ref{comp-conj}),
the Hermitian conjugate of a component  of an irreducible tensor operator with rank $\mu$ is still in the subspace  spanned by rank $\mu$ irreducible tensor operators. This implies that this subspace has a basis which is formed only from Hermitian operators. This together with the fact that  the  density operator $\rho$ itself  is a Hermitian operator imply that the components of $\rho$ for modes  with rank less than or equal $2l$ is uniquely specified by at most  $(2l+1)^{2}-1$ \emph{real} parameters. This completes the proof of theorem \ref{prop-gen-spin-l} in the case of measurements.
 
 Proof in the case of channels  follows  in the same way. The only difference is that the set of all possible modes that a  quantum channel acting on the system with Hilbert space $\mathcal{H}$ can have is determined by irreps which show up in the representation
$$g\rightarrow U(g)\otimes \bar{U}(g)\otimes U(g)\otimes \bar{U}(g)$$
of SO(3). Now, since the highest angular momentum which shows up in the representation $g\rightarrow U(g)$ is $l$, then the highest angular momentum which shows up in the above representation is $4l$. So an arbitrary channel acting on $\mathcal{H}$ can have mode in the set $\{(\mu,m): \ \mu\le 4l\}$. So, from lemma \ref{lem-mode-mode} to specify the performance of the quantum reference frame for simulating channels acting on this space we need to specify the components of $\rho$ for all modes with rank less than or equal $4l$. The rest of argument follows exactly the same as the  argument for the case of measurements.

\subsection{Proof of corollary \ref{cor-prop-gen-spin-l} }

We present the proof for the case of measurements. The proof for the case channels follows exactly in the same way. 

From theorem \ref{prop-gen-spin-l}  we know that to specify the performance of state $\rho$ as a quantum reference frame we need to specify all components of $\rho$ for all modes $\{(\mu,m): 1 \le \mu\le 2l\}$.  

Without loss of generality we assume state $\rho$ is invariant under rotations around $\hat{z}$. Now  for each mode $(\mu,m\neq 0)$, the corresponding component of the irreducible tensor operator basis, i.e. $T_{m}^{(\mu)}$ is not invariant under rotation around $\hat{z}$. Then, it follows that for all modes $(\mu,m\neq 0)$ the component of $\rho$ in those modes are zero, i.e. 
$$\forall (\mu,m\neq 0):\ \rho^{(\mu,m)}\equiv \text{tr}(\rho T_{m}^{(\mu)}{ }^{\dag})=0.$$
So we conclude that if the state $\rho$ is invariant under rotation around $\hat{z}$, then to specify its as a quantum reference frame we only need to specify its components in  modes $\{(\mu,0): 1 \le \mu\le 2l\}$. Now using Eq. (\ref{Clebsh-Gordon-tens}) we can 
easily show that the subspace spanned by 
$$\{T_{m=0}^{(\mu)}: 1 \le \mu\le 2l\} $$
is the same as the subspace spanned by 
$$\{ \left(T_{m=0}^{(1)}\right)^{k}: 1 \le k\le 2l\} $$
To see this we use  Eq. (\ref{Clebsh-Gordon-tens}) to decompose the product of irreducible tensor operators to the sum of irreducible tensor operators. Then Eq. (\ref{Clebsh-Gordon-tens}) implies that 
the problem of decomposing $ \left(T_{m=0}^{(1)}\right)^{k}$ to irreducible tensor operators is exactly equivalent to the problem of decomposing state $|j=1,m=0\rangle^{\otimes k}$ to irreps of SO(3). It follows that that  i) $ \left(T_{m=0}^{(1)}\right)^{k}$ has a nonzero component in mode $(\mu=k,0)$ and  ii) $ \left(T_{m=0}^{(1)}\right)^{k}$ does not have any  nonzero component in modes $(\mu>k,0)$.

So it follows that the span of $\{T_{m=0}^{(\mu)}: 1 \le \mu\le 2l\} $ is the same as span of $\{ \left(T_{m=0}^{(1)}\right)^{k}: 1 \le k\le 2l\} $. So to specify all the components of $\rho$ in modes  $\{(\mu,0): 1 \le \mu\le 2l\}$ one can specify all the moments of $\{\text{tr}(\rho \left(T_{m=0}^{(1)}\right)^{k}): 1 \le k\le 2l\}$ or equivalently the moments $\{\text{tr}(\rho L^{k}_{z}): 1 \le k\le 2l\}$. This completes the proof of corollary   \ref{cor-prop-gen-spin-l}.

\appendix




 




\bibliography{bibfile}


\end{document}